\documentclass[11pt,letterpaper]{article}
\usepackage[utf8]{inputenc}
\usepackage[margin=1in]{geometry}
\usepackage{palatino}
\usepackage{macros}

\usepackage{enumerate}
\usepackage{wasysym}

\usepackage[dvipsnames]{xcolor}
\hypersetup{
	colorlinks=true,
	pdfpagemode=UseNone,
    citecolor=OliveGreen,
    linkcolor=NavyBlue,
    urlcolor=Magenta,
	pdfstartview=FitW
}

\theoremstyle{plain}
\newtheorem{thm}{Theorem}
\newtheorem{lem}[thm]{Lemma}
\newtheorem{coro}[thm]{Corollary}
\newtheorem{prop}[thm]{Proposition}
\newtheorem{clm}[thm]{Claim}
\newtheorem{ex}[thm]{Example}

\theoremstyle{definition}
\newtheorem{defn}[thm]{Definition}

\theoremstyle{remark}
\newtheorem{rmk}[thm]{Remark}

\makeatletter
 \newcommand{\linkdest}[1]{\Hy@raisedlink{\hypertarget{#1}{}}}
\makeatother

\newcommand{\optheta}[1]{\lambda_{#1} \frac{\partial}{\partial \lambda_{#1}}}

\newcommand{\CovM}{K}
\newcommand{\grad}{\nabla}
\newcommand{\marginal}{M}

\newcommand{\mybf}[1]{\boldsymbol{#1}}

\newcommand{\sra}{\text{\scriptsize{~$\rightarrow$~}}}
\newcommand{\seq}{\text{\scriptsize{~$=$~}}}

\newcommand{\Influence}{\mathcal{I}}
\newcommand{\II}{\mathcal{I}}
\newcommand{\TSAW}{T_{\textsc{saw}}}
\newcommand{\sL}{{\sigma_\Lambda}}

\newcommand{\eps}{\epsilon}
\newcommand{\NP}{\mathsf{NP}}
\newcommand{\RP}{\mathsf{RP}}
\newcommand{\fpras}{\mathsf{FPRAS}}

\newcommand{\fptas}{\mathsf{FPTAS}}

\newcommand{\even}{\mathsf{even}}
\newcommand{\odd}{\mathsf{odd}}

\title{Rapid Mixing of Glauber Dynamics up to Uniqueness via~Contraction}
\author{
	Zongchen Chen\thanks{Georgia Institute of Technology. Email: chenzongchen@gatech.edu. 
	Research supported in part by NSF grant CCF-2007022.} 
	\and Kuikui Liu\thanks{University of Washington. Email: liukui17@cs.washington.edu.
	Research supported in part by NSF grant CCF-1907845 and ONR-YIP grant N00014-17-1-2429.} 
	\and Eric Vigoda\thanks{University of California, Santa Barbara. Email: vigoda@ucsb.edu. 
	Research supported in part by NSF grant CCF-2007022.} 
}
\date{\today}

\begin{document}
\maketitle

\begin{abstract}

For general antiferromagnetic $2$-spin systems, including the hardcore model on weighted independent sets and the
antiferromagnetic Ising model, 
there is an $\fptas$ for the partition function on graphs of maximum degree $\Delta$ when 
the infinite regular tree lies in the uniqueness region by Li et al. (2013).  Moreover, in the tree non-uniqueness region,
Sly (2010) showed that there is no $\fpras$ to estimate the partition function unless $\NP=\RP$.
The algorithmic results follow from the correlation decay approach due to Weitz (2006)
or the polynomial interpolation approach developed by Barvinok (2016). However the running time is only
polynomial for constant $\Delta$.  For the hardcore model, recent work of Anari et al. (2020) establishes rapid mixing 
of the simple single-site Markov chain known as the Glauber dynamics in the tree uniqueness 
region.  Our work simplifies their analysis of the Glauber dynamics by considering the total pairwise influence of a fixed vertex $v$ on other vertices, as opposed to the total influence of other vertices on $v$, thereby extending their work to
all 2-spin models and improving the mixing time.

More importantly our proof ties together the three disparate algorithmic approaches: 
we show that contraction of the so-called tree recursions with a suitable potential function, which
is the primary technique for establishing efficiency of Weitz's correlation decay approach and
Barvinok's polynomial interpolation approach, also establishes rapid mixing of the Glauber dynamics. We emphasize that 
this connection holds for all 2-spin models (both antiferromagnetic and ferromagnetic), and existing proofs for the correlation decay or
polynomial interpolation approach immediately imply rapid mixing of the Glauber dynamics. 
Our proof utilizes that the graph partition function is a divisor of
the partition function for Weitz's self-avoiding walk tree. This fact leads to new tools for the analysis of the
influence of vertices, and may be of
independent interest for the study of complex zeros. 
\end{abstract}

\thispagestyle{empty}

\newpage

\setcounter{page}{1}

\section{Introduction}
A remarkable connection has been established between the computational complexity of 
approximate counting problems in general graphs of maximum degree $\Delta$ and the 
statistical physics phase transition on infinite, regular trees of degree $\Delta$ (or up to $\Delta$ in the more general case). 
This connection holds for 2-state antiferromagnetic spin systems -- the hardcore model on independent sets and the 
Ising model are the most interesting examples of such systems.

Given an $n$-vertex graph $G = (V, E)$, configurations of the 2-spin model are the $2^n$ assignments of spins ${0, 1}$ to the vertices.
A 2-spin system is defined by three parameters: edge weights $\beta,\gamma>0$ and a vertex weight $\lambda>0$.
Edge parameter $\beta$ controls the (relative)
strength of interaction between neighboring $1$-spins, $\gamma$ corresponds to neighboring $0$-spins, and $\lambda$ is the external field
applied to vertices with $1$-spins.   

Every spin configuration $\sigma \in \{0,1\}^V$ is assigned a weight
$$ w_G(\sigma) = \beta^{m_1(\sigma)} \gamma^{m_0(\sigma)} \lambda^{n_1(\sigma)}, $$ 
where, for spin $s\in\{0,1\}$, $m_s(\sigma) = \#\{uv \in E : \sigma_u = \sigma_v = s\}$ is the number of monochromatic edges with spin $s$, 
and $n_1(\sigma) = \#\{v \in V : \sigma_v = 1\}$ is the number of vertices with spin $1$
(as is standard, the parameters are normalized so we can avoid two additional parameters).
The Gibbs distribution over spin configurations is given by
$\mu_G(\sigma) = \frac{w_G(\sigma)}{Z_G(\beta,\gamma,\lambda)}, $
where  
   $ Z_{G}(\beta,\gamma,\lambda) = \sum_{\sigma \in \{0,1\}^V} \beta^{m_1(\sigma)} \gamma^{m_0(\sigma)} \lambda^{n_1(\sigma)}$ 
is the partition function. 

There are two examples of particular interest: the hardcore model and the Ising model.
When $\beta=0$ and $\gamma=1$ then the only configurations with non-zero weight are 
independent sets of $G$ and the weight of an independent set $\sigma$ is $w(\sigma)=\lambda^{|\sigma|}$; this example is known as
the {\em hardcore model} where the parameter $\lambda$ corresponds to the fugacity.

In the case $\beta = \gamma$ then the important quantity is the total number of monochromatic edges 
$m(\sigma) = m_0(\sigma) + m_1(\sigma)$ and the weight of a configuration $\sigma$ is 
$w(\sigma)=\beta^{m(\sigma)}\lambda^{n_1(\sigma)}$; this is the classical {\em Ising model} where the parameter $\beta$ 
corresponds to the inverse temperature and $\lambda$ is the external field ($\lambda=1$ means no external field).
Note, when $\beta>1$ then the model is {\em ferromagnetic} as neighboring vertices prefer to have the same spin,
and $\beta<1$ is the {\em antiferromagnetic} Ising model.
In the general $2$-spin system, the model is ferromagnetic when $\beta\gamma>1$ and antiferromagnetic when $\beta\gamma<1$. (When $\beta\gamma = 1$ we get a trivial product distribution.)

The fundamental algorithmic tasks are to sample from the Gibbs distribution and to estimate the partition function.  
For the approximate sampling problem we are given a graph $G$ and an $\epsilon>0$ and our goal is to generate a sample
from a distribution $\pi$ which is within total variation distance $\leq\eps$ of the Gibbs distribution $\mu_G$ in
time $\poly(n,\log(1/\eps))$.   An efficient approximate sampling algorithm implies an $\fpras$ (fully-polynomial
randomized approximation scheme) for the approximate 
counting problem~\cite{JVV,SVV}.   Recall, given an $n$-vertex graph $G$, and $\eps,\delta>0$, an $\fpras$ 
outputs a $(1\pm\eps)$-approximation of $Z_G$ with probability $\geq 1-\delta$ in time $\poly(n,1/\eps,\log(1/\delta))$, whereas
an $\fptas$ is the deterministic analog (i.e., $\delta=0$).

A standard approach to the approximate sampling problem is the Markov Chain Monte Carlo (MCMC) method; in
fact there is a simple Markov chain known as the {\em Glauber dynamics}.  The Glauber dynamics works as follows: from
a configuration $X_t$ at time $t$, choose a random vertex $v$, we then set $X_{t+1}(w)=X_t(w)$ for all $w\neq v$, and 
finally we choose $X_{t+1}(v)$ from the conditional distribution of $\mu(\sigma_v|\sigma_w=X_{t+1}(w) \mbox{ for all } w\neq v)$.
For the case of the hardcore model, then $X_{t+1}(v)$ is set to occupied (i.e., spin $1$) with probability $\lambda/(1+\lambda)$ if
no neighbors are currently occupied, and otherwise it is set to unoccupied.  

It is straightforward to verify that the Glauber dynamics is ergodic with the Gibbs distribution as the unique stationary distribution.
The {\em mixing time} 
is the minimum number of steps to guarantee, from the worst initial state $X_0$, that the distribution of $X_t$ is within
total variation distance $\leq 1/4$ of the Gibbs distribution.  The goal is to prove that the mixing time is polynomial in $n$, in which
case the chain is said to be {\em rapidly mixing}.

For the case of the ferromagnetic Ising model (with or without an external field), 
a classical result of Jerrum and Sinclair~\cite{JS:ising} gives an $\fpras$ for all graphs via the MCMC method.
This is the only case with an efficient algorithm for general graphs.  
For antiferromagnetic 2-spin models the picture is closely tied to statistical physics phase transitions on the regular tree.

The uniqueness/non-uniqueness phase transition is nicely illustrated for the case of the hardcore model.
Consider the infinite $\Delta$-regular tree $T$ rooted at $r$, and let $T_h$ denote the tree truncated at the first $h$ levels.
This phase transition captures whether the configuration at the leaves of $T_h$ ``influences'' the root, in the limit $h\rightarrow\infty$.
For the hardcore model we can consider even height trees (corresponding to the all even boundary condition) versus
odd height trees.
Let $p_h$ denote the marginal probability that the root is occupied in the Gibbs distribution $\mu_{T_h}$.
Let $p_{\even} = \lim_{h\rightarrow\infty} p_{2h}$ and $p_{\odd} = \lim_{h\rightarrow\infty} p_{2h+1}$.
We say that {\em tree uniqueness} holds if $p_{\even} = p_{\odd}$ and {\em tree non-uniqueness} holds if they are not equal.
For all $\Delta\geq 3$ there exists a critical fugacity $\lambda_c(\Delta) = (\Delta-1)^{\Delta-1}/(\Delta-2)^{\Delta})$~\cite{Kelly},
where tree uniqueness holds iff $\lambda\leq\lambda_c(\Delta)$.

The remarkable connection is that an algorithmic phase transition for general graphs of maximum degree $\Delta$ occurs
at this same tree critical point.
For all constant $\Delta$, all $\delta>0$, all $\lambda<(1-\delta)\lambda_c(\Delta)$, all graphs of maximum degree $\Delta$,
\cite{Wei06} presented an $\fptas$ for approximating the partition function.  On the other side,
for all $\delta>0$, all $\lambda>(1+\delta)\lambda_c(\Delta)$,~\cite{Sly10,SS14,GSV16} proved that, 
unless $\NP=\RP$, there is no $\fpras$ for estimating the partition function.

One important caveat is that the running time of Weitz's algorithm is $(n/\eps)^{C\log\Delta}$ where the approximation factor is
$(1\pm\eps)$ and the constant $C$ depends polynomially on the gap $\delta$ (recall, $\lambda<(1-\delta)\lambda_c$).  
Weitz's correlation decay algorithm was extended to 
the antiferromagnetic Ising model in the tree uniqueness region by Sinclair et al.~\cite{SST14}, 
and to all antiferromagnetic 2-spin systems in the corresponding
tree uniqueness region (as we detail below) by Li, Lu, and Yin~\cite{LLY13}. 

An intriguing new algorithmic approach was presented by Barvinok~\cite{Bar17} and refined by Patel and Regts~\cite{PR17},
utilizing the absence of zeros of the partition function in the complex plane to efficiently approximate a suitable transformation of
the logarithm of the partition function using Taylor approximation.  This polynomial interpolation approach was shown to be efficient in the same
tree uniqueness region as for Weitz's result by Peters and Regts~\cite{PR19}, although the exponent in the running time depends
exponentially on $\Delta$.

It was long conjectured that the simple Glauber dynamics is rapidly mixing in the tree uniqueness region.  
This was recently proved by Anari, Liu, and Oveis Gharan \cite{ALO20}; they proved, for all $\delta>0$, the mixing time is $n^{O(\exp(1/\delta))}$
whenever $\lambda<(1-\delta)\lambda_c(\Delta)$.    We improve this result.  First, we improve the 
mixing time from $n^{O(\exp(1/\delta))}$ to $n^{O(1/\delta)}$ as detailed in the following theorem.

\begin{thm}[Hardcore model]
\label{thm:main-hardcore}
Let $\Delta \ge 3$ be an integer and $\delta \in (0,1)$. 
For every $n$-vertex graph $G$ of maximum degree $\Delta$ and every $0< \lambda \le (1-\delta) \lambda_c(\Delta)$, 
the mixing time of the Glauber dynamics for the hardcore model on $G$ with fugacity $\lambda$ is $O(n^{2+32/\delta})$. 
\end{thm}
This bound is optimal barring further improvements in the local-to-global arguments from \cite{AL20}. 
Our improved result follows from a simpler, cleaner proof approach 
which enables us to extend our result to a wide variety of 2-spin models, matching the key results for the correlation decay algorithm with 
vastly improved running times.

Our proof approach unifies the three major algorithmic tools for approximate counting: correlation decay, polynomial interpolation, and MCMC.
Most known results for both correlation decay and polynomial interpolation approach are proved by showing contraction of a suitably defined potential function on
the so-called tree recursions; the tree recursions arise as a result of Weitz's self-avoiding walk tree that we will describe
in more detail later in this paper.  A recent work of Shao and Sun~\cite{SS19} unifies these two approaches by showing
that the contraction which is normally used to prove efficiency of the correlation decay algorithm,
also implies (under some additional analytic conditions) that the polynomial interpolation approach is efficient.

Here we prove that this same contraction of a potential function also implies rapid mixing of the
Glauber dynamics, with our improved running time that is independent of $\Delta$; see \cref{defn:potential}
and \cref{thm:contraction-implies-mixing}
for a detailed statement.  Our proof utilizes 
several new tools concerning Weitz's self-avoiding walk tree, which are detailed in \cref{sec:intro-proof}.  
In particular, we show that the partition function of a graph $G$ divides the partition function of Weitz's self-avoiding walk tree; see \cref{lem:Tsaw-G}.
This result is potentially of independent interest for establishing absence of zeros for the partition function with complex parameters, as
it enables one to consider the self-avoiding walk tree.  
This result also yields a new, useful equivalence for bounding the
influence in a graph in terms of the self-avoiding tree, which strengthens the previously known connection by Weitz~\cite{Wei06}; see \cref{lem:Tsaw-G} for details.

As an easy consequence we obtain 
rapid mixing for the Glauber dynamics for the antiferromagnetic Ising model in the tree uniqueness region.
In terms of the edge activity, the two critical points for the Ising model on the $\Delta$-regular tree are at
$\beta_c(\Delta) = \frac{\Delta-2}{\Delta}$ and $\overline{\beta}_{c}(\Delta) = \frac{1}{\beta_{c}(\Delta)} = \frac{\Delta}{\Delta-2}$; the first lies in the antiferromagnetic regime, while the second lies in the ferromagnetic regime. If $\beta_{c}(\Delta) < \beta < \overline{\beta}_{c}(\Delta)$, then uniqueness holds for all external field $\lambda$ on the $\Delta$-regular tree.


As mentioned earlier, for the ferromagnetic Ising model, an $\fpras$ was known for general graphs \cite{JS:ising}. Furthermore, Mossel and Sly \cite{MS13} proved $O(n\log{n})$ mixing time of the Glauber dynamics
for the ferromagnetic Ising model when $1 \leq \beta < \overline{\beta}_{c}(\Delta)$.  However, rapid mixing for the antiferromagnetic Ising model in the tree uniqueness region was not known.

We provide the following mixing result for the case $\beta > \beta_{c}(\Delta)$. Note, when $\beta\leq\beta_c$ there is an additional uniqueness region for certain values of the external field $\lambda$; this region is covered by~\cref{thm:main-2-spin}.


\begin{thm}[Antiferromagnetic Ising Model]
\label{thm:main-Ising}
Let $\Delta \ge 3$ be an integer and $\delta \in (0,1)$. 
Assume that $1 > \beta \ge \beta_{c}(\Delta) + \delta(1 - \beta_{c}(\Delta))$ and $\lambda > 0$. 
Then for every $n$-vertex graph $G$ of maximum degree $\Delta$, the mixing time of the Glauber dynamics for the Ising model on $G$ with edge weight $\beta$ and external field $\lambda$ is $O(n^{2+1.5/\delta})$.
\end{thm}

%
%

Our results for the hardcore and Ising models fit within a larger framework of general antiferromagnetic 2-spin systems.
Recall that the antiferromagnetic case is when $\beta\gamma<1$.

For general 2-spin systems the appropriate tree phase transition is more complicated as there are models where
the tree uniqueness threshold is not monotone in $\Delta$.  Hence the appropriate notion is ``up-to-$\Delta$ uniqueness'' as considered by~\cite{LLY13}.
Roughly speaking, we say uniqueness with gap $\delta \in (0,1)$ holds on the $d$-regular tree if for every integer $\ell \geq 1$, all vertices at distance $\ell$ from the root have total ``influence'' $\lesssim (1 - \delta)^{\ell}$ on the marginal of the root. We say up-to-$\Delta$ uniqueness with gap $\delta$ holds if uniqueness with gap $\delta$ holds on the $d$-regular tree for all $1 \leq d \leq \Delta$; see \cref{sec:preliminaries} for the precise definition.

Both \cref{thm:main-hardcore} and \cref{thm:main-Ising} are corollaries of the following general rapid mixing result which holds for 
general antiferromagnetic $2$-spin systems in the entire tree uniqueness region.

\begin{thm}[General antiferromagnetic $2$-spin system]
\label{thm:main-2-spin}
Let $\Delta \ge 3$ be an integer and $\delta \in (0,1)$. 
Let $\beta,\gamma,\lambda$ be reals such that $0\le \beta \le \gamma$, $\gamma>0$, $\beta\gamma <1$ and $\lambda>0$. 
Assume that the parameters $(\beta,\gamma,\lambda)$ are up-to-$\Delta$ unique with gap $\delta$. 
Then for every $n$-vertex graph $G$ of maximum degree $\Delta$, the mixing time of the Glauber dynamics for the antiferromagnetic $2$-spin system on $G$ with parameters $(\beta,\gamma,\lambda)$ is $O(n^{2+72/\delta})$. 
\end{thm}

We also match existing correlation decay results \cite{GL18, SS19} for ferromagnetic 2-spin models; see \cref{sec:ferro} for results, and \cref{sec:ferroproofs} for proofs.

\subsection{Mixing by the potential method}
\label{sec:mixing-potential}
The tree recursion is very useful in the study of approximating counting. 
Consider a tree rooted at $r$. 
Suppose that $r$ has $d$ children, denoted by $v_1,\dots,v_d$. 
For $1\le i\le \Delta_i$ we define $T_{v_i}$ to be the subtree of $T$ rooted at $v_i$ that contains all descendant of $v_i$. 
Let $R_r = \mu_T(\sigma_r \seq 1) / \mu_T(\sigma_r \seq 0)$ denote the marginal ratio of the root, 
and $R_{v_i} = \mu_{T_{v_i}}(\sigma_{v_i} \seq 1) / \mu_{T_{v_i}}(\sigma_{v_i} \seq 0)$ for each subtree. 
The \emph{tree recursion} is a formula that computes $R_r$ given $R_{v_1},\dots,R_{v_d}$, due to the independence of $T_{v_i}$'s. 
More specifically, we can write $R_r = F_d(R_{v_1},\dots, R_{v_d})$ where $F_d:[0,+\infty]^d \to [0,+\infty]$ is a multivariate function such that for $(x_1,\dots,x_d) \in [0,+\infty]^d$, 
\[
F_d(x_1,\dots,x_d) = \lambda \prod_{i=1}^d \frac{\beta x_i + 1}{x_i + \gamma}. 
\]

In this paper, however, we pay particular interest in the log of marginal ratios. 
The reason is that we will carefully study the \emph{pairwise influence matrix} $\II_G$ of the Gibbs distribution $\mu_G$, introduced in \cite{ALO20} and defined as for every $r,v \in V$ 
\begin{equation*}
    \II_G(r \sra v) = \mu_G(\sigma_v \seq 1 \mid \sigma_r \seq 1) - \mu_G(\sigma_v \seq 1 \mid \sigma_r \seq 0). 
\end{equation*}
In \cite{ALO20}, the authors show that if the maximum eigenvalue of $\II_G$ is bounded appropriately, 
then the Glauber dynamics is rapid mixing. 
One crucial observation we make in this paper is that the influence $\II_G(r \sra v)$ of $r$ on $v$ can be viewed as the derivative of $\log R_r$ with respect to the log external field at $v$ (see \cref{lem:2spin-sys-property}). 
Thus, it is more convenient for us to work with the log ratios. 
To this end, we rewrite the tree recursion as $\log R_v = H_d(\log R_{v_1}, \dots, \log R_{v_d})$ where $H_d: [-\infty, +\infty]^d \to [-\infty, +\infty]$ is a function such that for $(y_1,\dots,y_d)\in [-\infty, +\infty]^d$, 
\[
H_d(y_1,\dots,y_d) = \log \lambda + \sum_{i=1}^d \log \left( \frac{\beta e^{y_i} + 1}{e^{y_i} + \gamma} \right). 
\]
Observe that $H = \log \circ F \circ \exp$. 
Moreover, we define 
\[
h(y) = - \frac{(1-\beta\gamma) e^y}{(\beta e^y + 1)(e^y + \gamma)} 
\]
for $y\in[-\infty,+\infty]$, 
so that $\frac{\partial}{\partial y_i} H_d(y_1,\dots,y_d) = h(y_i)$ for each $i$. 

To prove our main results, we use the potential method, which has been widely used to establish the decay of correlation. 
By choosing a suitable potential function for the log ratios, 
we show that the total influence from a given vertex decays exponentially with the distance, and thus establish 
rapid mixing of the Glauber dynamics. 
Let us first specify our requirements on the potential.
For every integer $d\ge 0$, we define a bounded interval $J_d$ which contains all log ratios at a vertex of degree $d$. 
More specifically, we let $J_{d} = \wrapb{\log(\lambda \beta^{d}), \log(\lambda / \gamma^{d})}$ when $\beta\gamma < 1$, and $J_{d} = \wrapb{\log(\lambda / \gamma^{d}), \log(\lambda \beta^{d})}$ when $\beta\gamma > 1$. Furthermore, define $J = \bigcup_{d = 0}^{\Delta-1} J_{d}$ to be the interval containing all log ratios with degree less than $\Delta$. 
\begin{defn}[$(\alpha,c)$-Potential function]
\label{defn:potential}
Let $\Delta \ge 3$ be an integer. 
Let $\beta,\gamma,\lambda$ be reals such that $0\le \beta \le \gamma$, $\gamma>0$ and $\lambda>0$. 
Let $\Psi:[-\infty,+\infty] \to (-\infty,+\infty)$ be a differentiable and increasing function with image $S = \Psi[-\infty,+\infty]$ and derivative $\psi = \Psi'$. 
For any $\alpha\in(0,1)$ and $c>0$, 
we say $\Psi$ is an \emph{$(\alpha,c)$-potential function} with respect to $\Delta$ and $(\beta, \gamma, \lambda)$ if it satisfies the following conditions:
\begin{enumerate}
\item \linkdest{cond:contraction}(Contraction) For every integer $d$ such that $1\le d<\Delta$ and every $(\tilde{y}_1,\dots,\tilde{y}_d) \in S^d$, we have
\[
\norm{\grad H_{d}^{\Psi}(\tilde{y}_1,\dots,\tilde{y}_d)}_1 = \sum_{i=1}^d \frac{\psi(y)}{\psi(y_i)} \cdot |h(y_i)| \le 1-\alpha
\]
where $H_{d}^{\Psi} = \Psi \circ H_{d} \circ \Psi^{-1}$, $y_i = \Psi^{-1}(\tilde{y}_i)$ for $1\le i \le d$, and $y = H_{d}(y_1,\dots,y_d)$. 
\item \linkdest{cond:boundedness}(Boundedness) For every $y_{1},y_{2} \in J$, we have 
\[
\frac{\psi(y_{2})}{\psi(y_{1})} \cdot \abs{h(y_{1})} \leq \frac{c}{\Delta}.
\]
\end{enumerate}
\end{defn}

In the definition of $(\alpha,c)$-potential, one should think of $y$ as the log marginal ratio at a vertex and the potential function is of $\log R$. 
The following theorem establishes rapid mixing of the Glauber dynamics given an $(\alpha,c)$-potential function. 


\begin{thm}
\label{thm:contraction-implies-mixing}
Let $\Delta \ge 3$ be an integer. 
Let $\beta,\gamma,\lambda$ be reals such that $0\le \beta \le \gamma$, $\gamma>0$ and $\lambda>0$. 
Suppose that there is an $(\alpha,c)$-potential with respect to $\Delta$ and $(\beta,\gamma,\lambda)$ for some $\alpha \in (0,1)$ and $c>0$. 
Then for every $n$-vertex graph $G$ of maximum degree $\Delta$, 
the mixing time of the Glauber dynamics for the $2$-spin system on $G$ with parameters $(\beta,\gamma,\lambda)$ is $O(n^{2+c/\alpha})$. 
\end{thm}


We outline our proofs in \cref{sec:intro-proof}. 
Note that in both \cref{defn:potential} and \cref{thm:contraction-implies-mixing}, the constant $c$ is allowed to depend on the maximum degree $\Delta$ and parameters $(\beta,\gamma,\lambda)$ in general. 
For example, a straightforward black-box application of the potential in \cite{LLY13} would give $c = \Theta(\Delta)$ for the \hyperlink{cond:boundedness}{Boundedness} condition, resulting in $n^{\Theta(\Delta)}$ mixing. 
However, this is undesirable for graphs with potentially unbounded degrees. 
One of our contributions is that we show the \hyperlink{cond:boundedness}{Boundedness} condition holds for a universal constant $c$ \emph{independent} of $\Delta$ and $(\beta,\gamma,\lambda)$. 
Thus, our mixing time is $O(n^{2+c/\delta})$ with no parameters in the exponent except for $1/\delta$. 

In \cref{sec:antiferroweighted}, we give a slightly more general definition of $(\alpha,c)$-potentials, 
which relaxes the \hyperlink{cond:boundedness}{Boundedness} condition, and is necessary for our analysis of antiferromagnetic $2$-spin systems with $0 \le \beta < 1 < \gamma$. 
\cref{thm:contraction-implies-mixing} still holds for this larger class of potentials. 

We remark that in all previous works of the potential method, results and proofs are always presented in terms of $F_d$, the tree recursion of $R$, and $\Phi$, a potential function of $R$. 
In fact, our results can also be translated into the language of $(F_d,\Phi)$. 
To see this, since $H_d = \log \circ F_d \circ \exp$, it is straightforward to check that $H_d^\Psi = \Psi \circ H_{d} \circ \Psi^{-1} = \Phi \circ F_d \circ \Phi^{-1} = F_d^\Phi$ if we pick $\Phi = \Psi \circ \log$, 
and thereby $\grad H_d^\Psi = \grad F_d^\Phi$.  
This implies that the \hyperlink{cond:contraction}{Contraction} condition in \cref{defn:potential} holds for $(H_d,\Psi)$ if and only if the corresponding contraction condition holds for $(F_d,\Phi)$. 
The \hyperlink{cond:boundedness}{Boundedness} condition can also be stated equivalently for $(F_d,\Phi)$. 
Nevertheless, in this paper we choose to work with $(H_d,\Psi)$ for the following two reasons. 
First, as mentioned earlier, the fact that $\II_G(r \sra v)$ is a derivative of $\log R_r$ makes it natural to consider the tree recursion for the log ratios. Indeed, it is easier and cleaner to present our results and proofs using $(H_d,\Psi)$ directly rather than switching to $(F_d,\Phi)$. 
Second, the potential function $\Psi$ we will use is obtained from the exact potential $\Phi$ in \cite{LLY13}, by the transformation $\Psi = \Phi \circ \exp$.\footnote{To be more precise, we also multiply a constant factor which only simplifies our calculation and does not matter much; also notice that \cite{LLY13} denotes the potential function by $\varphi$ and its derivative by $\Phi = \varphi'$.} 
It is intriguing to notice that the derivative of this potential is simply $\psi = \sqrt{|h|}$. 
Then the \hyperlink{cond:contraction}{Contraction} condition has a nice form: $\sum_{i=1}^d \sqrt{h(y) h(y_i)} \le 1-\alpha$; and the \hyperlink{cond:boundedness}{Boundedness} condition only involves an upper bound on $h(y)$. 
This seems to shed some light on the mysterious potential function $\Phi$ from \cite{LLY13}, and also indicates 
that $H_d$ is a meaningful variant of the tree recursion to consider. 
To add one more evidence, for a lot of cases (e.g., $\frac{\Delta-2}{\Delta} < \sqrt{\beta\gamma} < \frac{\Delta}{\Delta-2}$) where the potential $\Phi = \log$ is picked, that just means we can pick $\Psi$ to be the identity function and $H_d$ itself is contracting without any nontrivial potential.

\paragraph{Revision in July 2021.}
After the publication of this paper in FOCS 2020, a small error was found in \cite{LLY13} regarding descriptions of the uniqueness region for antiferromagnetic 2-spin systems. 
The error was fixed in the latest version of \cite{LLY13}. 
In this revision, we update corresponding results and proofs in \cref{sec:antiferroweighted} and \cref{sec:boundedness} that rely on the changes in \cite{LLY13}; in particular, \cref{lem:bound-Psi} is adjusted in accordance with the current description of uniqueness regions.
We remark that these changes are purely technical and do not affect the validity of our main results like \cref{thm:contraction-implies-mixing}.

\paragraph{Acknowledgments.} 
We would like to thank Shayan Oveis Gharan and Nima Anari for stimulating discussions. 
We also thank the anonymous referees for helpful comments and suggestions. 
We are grateful to Yitong Yin for communicating with us about the latest update of \cite{LLY13} and for providing helpful instructions on modifying statements and proofs of results in \cref{sec:boundedness}, particularly \cref{lem:bound-Psi}.

\section{Preliminaries}\label{sec:preliminaries}

\subsubsection*{Mixing time and spectral gap}

\noindent
Let $P$ be the transition matrix of an ergodic (i.e., irreducible and aperiodic) Markov chain on a finite state space $\Omega$ with stationary distribution $\mu$. 
Let $P^t(x_0,\cdot)$ denote the distribution of the chain after $t$ steps starting from $x_0 \in \Omega$. 
The \emph{mixing time} of $P$ is defined as
\[
T_\mathrm{mix}(P) = \max_{x_0 \in \Omega} \min \left\{t \ge 0: \norm{P^t(x_0,\cdot) - \mu(\cdot)}_{\mathrm{TV}} \le \frac{1}{4} \right\}. 
\]
We say $P$ is \emph{reversible} if $\mu(x) P(x,y) = \mu(y) P(y,x)$ for all $x,y\in\Omega$. 
If $P$ is reversible, then $P$ has only real eigenvalues which can be denoted by $1=\lambda_1 \ge \dots \ge \lambda_{|\Omega|} \ge -1$. 
The \emph{spectral gap} of $P$ is defined to be $1-\lambda_2$ and 
the \emph{absolute spectral gap} of $P$ is defined as $\lambda^*(P) = 1 - \max\{|\lambda_2|, |\lambda_{|\Omega|}|\}$. 
If $P$ is also positive semidefinite with respect to the inner product $\langle \cdot, \cdot\rangle_\mu$, then all eigenvalues of $P$ are nonnegative and thus $\lambda^*(P) = 1-\lambda_2$. 
Finally, the mixing time and the absolute spectral gap are related by
\begin{equation}\label{eq:rel-mixing}
T_\mathrm{mix}(P) \le \frac{1}{\lambda^*(P)} \log\left(\frac{4}{\min_{x\in\Omega} \mu(x)}\right). 
\end{equation}



\subsubsection*{Uniqueness}

\noindent
Let $\Delta \ge 3$ be an integer or $\Delta = \infty$. 
Let $\beta,\gamma,\lambda$ be reals such that $0\le \beta \le \gamma$, $\gamma>0$, $\beta\gamma <1$ and $\lambda>0$. 
For $1 \leq d < \Delta$, define 
\[
f_{d}(R) = \lambda \wrapp{\frac{\beta R + 1}{R + \gamma}}^{d}
\]
and denote the unique fixed point of $f_{d}$ by $R_d^*$. 
For $\delta \in(0,1)$, 
we say the parameters $(\beta,\gamma,\lambda)$ are \emph{up-to-$\Delta$ unique with gap $\delta$} if $|f'_d(R^*_d)| < 1- \delta$ for all $1 \leq d < \Delta$. 

\subsubsection*{Ratio and influence}

\noindent
Consider the $2$-spin system on a graph $G=(V,E)$. 
Let $\Lambda \subseteq V$ and $\sL \in \{0,1\}^\Lambda$. 
For all $v\in V \backslash \Lambda$, we define the \emph{marginal ratio} at $v$ to be
\[
R_G^\sL(v) 
= \frac{\mu_G(\sigma_v \seq 1 \mid \sL)}{\mu_G(\sigma_v \seq 0 \mid \sL)}. 
\]
For all $u,v \in V \backslash \Lambda$, we define the \emph{(pairwise) influence} of $u$ on $v$ by
\[
\II_G^\sL(u \sra v) = \mu_G(\sigma_v \seq 1 \mid \sigma_u \seq 1, \,\sL) - \mu_G(\sigma_v \seq 1 \mid \sigma_u \seq 0, \,\sL).
\]
Write $\II_G^\sL$ for the \emph{(pairwise) influence matrix} whose entries are given by $\II_G^\sL(u \sra v)$. 
\subsubsection*{Weitz's self-avoiding walk tree}

\noindent
Let $G=(V,E)$ be a connected graph and $r\in V$ be a vertex of $G$. 
The \emph{self-avoiding walk (SAW) tree} is defined as follows. 
Suppose that there is a total ordering of the vertex set $V$. 
A self-avoiding walk from $r$ is a path $r=v_0 - v_1 - \dots - v_\ell$ such that $v_i\neq v_j$ for all $0\le i < j\le \ell$. 
The SAW tree $T_{\textsc{saw}}(G,r)$ is a tree rooted at $r$, consisting of all self-avoiding walks $r=v_0 - v_1 - \dots - v_\ell$ with $\deg(v_\ell) = 1$, and those appended with one more vertex that closes the cycle (i.e., $r=v_0 - v_1 - \dots - v_\ell - v_i$ for some $0\le i \le \ell-2$ such that $\{v_\ell,v_i\} \in E$). 
Note that a vertex of $G$ might have many copies in the SAW tree, and the degrees of vertices are preserved except for leaves. 
See \cref{fig:TSAW} for an example. 

We can define a $2$-spin system on $T_{\textsc{saw}}(G,r)$ with the same parameters $(\beta,\gamma,\lambda)$, in which some of the leaves are fixed to a particular spin. 
More specifically, for a self-avoiding walk $r=v_0 - v_1 - \dots - v_\ell$ appended with $v_i$, we fix $v_i$ to be spin $1$ if $v_{i+1} < v_\ell$ with respect to the total ordering on $V$, and spin $0$ if $v_{i+1} > v_\ell$. 
For each $v\in V$ we denote the set of all free (unfixed) copies of $v$ in $\TSAW(G,r)$ by $\mathcal{C}_v$. 
For $\Lambda \subseteq V$ and a partial configuration $\sigma_\Lambda \in \{0,1\}^\Lambda$, we define the SAW tree with conditioning $\sL$ 
by assigning the spin $\sigma_v$ to every copy $\hat{v}$ of $v$ from $\mathcal{C}_v$ and removing all descendants of $\hat{v}$, for each $v\in \Lambda$. 
Note that in general, different copies of $v$ from $\mathcal{C}_v$ can receive different spin assignments. 
Finally, in the case that every vertex $v$ has a distinct field $\lambda_v$, all copies of $v$ from $\mathcal{C}_v$ will have the same field $\lambda_v$ in the SAW tree. 

\begin{figure}[t]
\centering
\includegraphics[width = 0.9\textwidth]{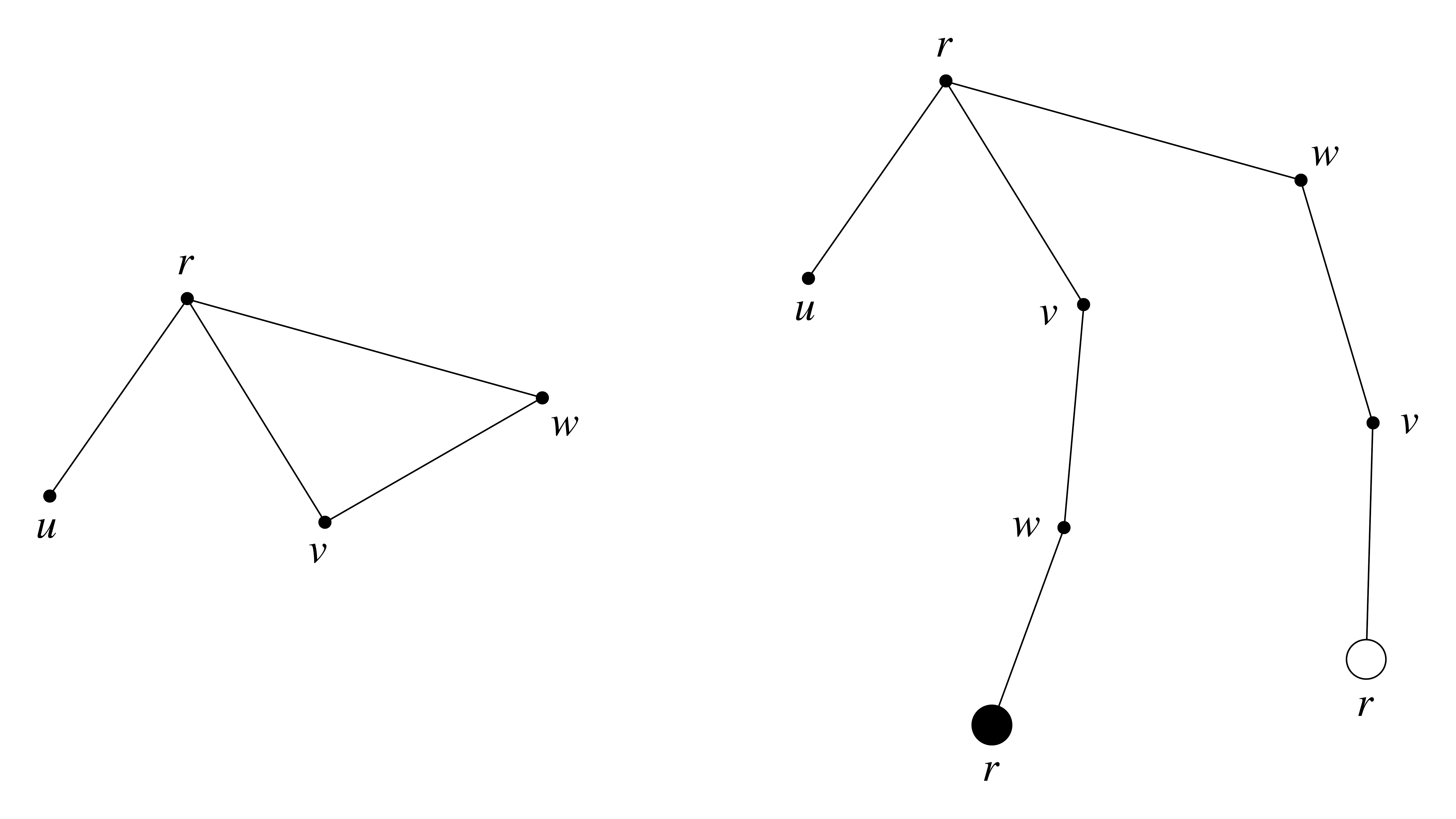}
\caption{A graph $G$ and the self-avoiding walk tree $\TSAW(G,r)$ rooted at $r$. Vertices with the same label in $\TSAW(G,r)$ are copies of the same vertex from $G$. ($\CIRCLE$/$\Circle$: fixed to spin $1$/$0$.)}
\label{fig:TSAW}
\end{figure}

\section{Proof outline for main results}
\label{sec:intro-proof}

\subsubsection*{Step 1 (\cite{ALO20}): Spectral Independence implies rapid mixing.}

\noindent
Our proof builds on \cite{ALO20} who showed that the Glauber dynamics for sampling from the hardcore distribution on graphs of maximum degree at most $\Delta$ mixes in $O(n^{\exp(O(1/\delta))})$ steps whenever $\lambda \leq (1 - \delta)\lambda_{c}(\Delta)$. 
One of the key ingredients of their proof is a notion they call spectral independence. \cite{ALO20} shows that the spectral independence property implies rapid mixing. 
Note that the diagonal entries of $\II_G^\sL$ are $1$, as opposed to $0$ in the original definition in \cite{ALO20}.

\begin{defn}[Spectral Independence \cite{ALO20}]
We say the Gibbs distribution $\mu_G$ on an $n$-vertex graph $G$ is \emph{$(\eta_{0}, \dots, \eta_{n-2})$-spectrally independent}, 
if for every $0\le k \le n-2$, $\Lambda\subseteq V$ of size $k$ and $\sigma_\Lambda \in \{0,1\}^\Lambda$, 
one has $\lambda_{\max}(\mathcal{I}_G^\sL) - 1 \leq \eta_k$. 
\end{defn}

\begin{thm}[\cite{ALO20}]
\label{thm:spectral-influence}
If $\mu$ is an $(\eta_{0},\dots,\eta_{n-2})$-spectrally independent distribution, then the Glauber dynamics for sampling from $\mu$ has spectral gap at least 
$$\frac{1}{n}\, \prod_{i=0}^{n-2} \wrapp{1 - \frac{\eta_{i}}{n-i-1}}. $$
\end{thm}
Our primary goal now is to bound the maximum eigenvalue of $\II_G^\sL$. 

\subsubsection*{Step 2: Self-avoiding walk trees preserve influences.}

\noindent
From standard linear algebra, we know that the maximum eigenvalue of $\II_G^\sL$ is upper bounded by both the $1$-norm $\norm{\II_G^\sL}_1 = \max_{r\in V} \sum_{v\in V} |\II_G^\sL(v \sra r)|$, which corresponds to total influences on a vertex $r$, and the infinity-norm $\norm{\II_G^\sL}_\infty = \max_{r\in V} \sum_{v\in V} |\II_G^\sL(r \sra v)|$, corresponding to total influences of $r$. 
In \cite{ALO20} the authors use $\norm{\II_G^\sL}_1$ as an upper bound on $\lambda_{\max} (\II_G^\sL)$. 
Roughly speaking, they show that the sum of absolute influences on a fixed vertex $r$, is upper bounded by the maximum absolute influences on $r$ in the self-avoiding walk tree rooted at $r$, over all boundary conditions. 
Here in this paper, we will use $\norm{\II_G^\sL}_\infty$ to upper bound $\lambda_{\max}(\II_G^\sL)$ instead. In fact, much more is true if we look at the influences from $r$ in the self-avoiding tree. 
We show that for every vertex $v\in V$, the influence $\II_G^\sL(r \sra v)$ in $G$ is preserved in the self-avoiding walk tree $T = \TSAW(G,r)$ rooted at $r$, in the form of sum of influences $\II_T^\sL(r \sra \hat{v})$ over all copies $\hat{v}$ of $v$. 

The way we establish this fact is by viewing the partition function as a polynomial in $\lambda$. 
In fact, it will be useful to consider the more general case with an arbitrary external field $\lambda_v$ for every $v\in V$. 
Let $\mybf{\lambda} = \{\lambda_v: v\in V\}$ denote the fields. 
For $\Lambda \subseteq V$ and $\sL \in \{0,1\}^\Lambda$, the weight of $\sigma \in \{0,1\}^{V\backslash \Lambda}$ conditional on $\sL$ is defined to be
$w_G(\sigma \mid \sL) = \beta^{m_1(\sigma \mid \sL)} \gamma^{m_0(\sigma \mid \sL)} \prod_{v\in V\backslash \Lambda} \lambda_v^{\sigma_v}$ 
where $m_i(\cdot \mid \sL)$ is the number of $i$-$i$ edges with at least one endpoint in $V\backslash \Lambda$ for $i=0,1$. 
Furthermore, $Z_G^\sL = \sum_{\sigma \in \{0,1\}^{V\backslash \Lambda}} w_G(\sigma \mid \sL)$ is the partition function conditioned on $\sL$. 
We shall view $\beta$ and $\gamma$ as some fixed constants and think of $\mybf{\lambda}$ as $n = |V|$ variables. 
In this sense, we regard the weights $w_G(\sigma \mid \sL)$ as monomials in $\mybf{\lambda}$ and the partition function $Z_G^\sL$ as a polynomial in $\mybf{\lambda}$. Moreover, the marginal ratios $R_G^\sL(v)$ and the influences $\II_G^\sL(r \sra v)$ for $r,v\in V$ are all functions in $\mybf{\lambda}$. 
Our main result is that the partition function of $G$ divides that of $\TSAW(G,r)$ for each $r\in V$. 
From that, we show that the SAW tree preserves influences of the root, as well as re-establishing Weitz's celebrated result \cite{Wei06}, see \cref{lem:I-G-T}. 

\begin{lem}
\label{lem:Tsaw-G}
Let $G=(V,E)$ be a connected graph, $r\in V$ be a vertex and $\Lambda \subseteq V \backslash \{r\}$ such that $G\backslash \Lambda$ is connected. 
Let $T = T_{\textsc{saw}}(G,r)$ be the self-avoiding walk tree of $G$ rooted at $r$. 
Then for every $\sL \in \{0,1\}^\Lambda$, $Z_G^\sL$ divides $Z_T^\sL$. 
More precisely, there exists a polynomial $P_{G,r}^\sL = P_{G,r}^\sL(\mybf{\lambda})$ independent of $\lambda_r$ such that
\begin{align}\label{eq:TsawDivisibility}
    Z_T^\sL = Z_G^\sL \cdot P_{G,r}^\sL.
\end{align}

As a corollary, for each vertex $v\in V$, 
\begin{align}\label{eq:TsawGinf}
    \Influence_G^\sL(r\sra v) = \sum_{\hat{v} \in \mathcal{C}_v} \Influence_T^\sL(r\sra \hat{v}), 
\end{align}
where $\mathcal{C}_v$ is the set of all free (unfixed) copies of $v$ in $T$. 
\end{lem}
\begin{remark}\label{rem:divisibility}
We emphasize that for the purposes of bounding the total influence of a vertex in $G$, only \cref{eq:TsawGinf} of \cref{lem:Tsaw-G} is needed, which can be proved in a purely combinatorial fashion. However, we believe the divisibility property \cref{eq:TsawDivisibility} of the multivariate partition function of $G$ and its self-avoiding walk tree may be of independent interest.
\end{remark}

We note that a univariate version of the divisibility statement \cref{eq:TsawDivisibility} has already appeared in \cite{Ben18} for the hardcore model and \cite{LSS19} for the zero-field Ising model in the study of complex roots of the partition function. 
From \cref{lem:Tsaw-G}, we can get $\sum_{v\in V} |\II_G^\sL(r \sra v)| \le \sum_{v\in V_T} |\II_T^\sL(r \sra v)|$ for any fixed $r$. 
That means, we only need to upper bound the sum of all influences for trees, in order to get an upper bound on $\lambda_{\max}(\II_G^\sL)$. 

\subsubsection*{Step 3: Decay of influences given a good potential.}

\noindent
The tree recursion provides us a great tool for computing the (log) ratios of vertices recursively for trees. 
As we show in \cref{lem:2spin-sys-property}, the influence $\II_G^\sL(r \sra v)$ is in fact a version of derivative of the log marginal ratio at $r$. 
Thus, the tree recursion can be used naturally to relate these influences. 
We then apply the potential method, which has been widely used in literature to establish the decay of correlations (strong spatial mixing). 
The following lemma shows that the sum of absolute influences to distance $k$ has exponential decay with $k$, which can be thought of as the decay of pairwise influences. 

\begin{lem}
\label{lem:bound-for-tree}
If there exists an $(\alpha,c)$-potential function $\Psi$ with respect to $\Delta$ and $(\beta,\gamma,\lambda)$ where $\alpha \in (0,1)$ and $c>0$, then for every $\Lambda \subseteq V_T \backslash \{r\}$, $\sL \in \{0,1\}^\Lambda$ and all integers $k \geq 1$, 
\begin{equation*}
    \sum_{v \in L_{r}(k)} \abs{\mathcal{I}_{T}^\sL(r\sra v)} \leq c \cdot (1 - \alpha)^{k-1} 
\end{equation*}
where $L_r(k)$ denote the set of all free vertices at distance $k$ away from $r$. 
\end{lem}

\cref{thm:contraction-implies-mixing} is then proved by combining \cref{thm:spectral-influence}, \cref{lem:Tsaw-G} and \cref{lem:bound-for-tree}. 
We leave its proof to \cref{sec:proof-main}. 

\subsubsection*{Step 4: Find a good potential.}

\noindent
As our final step, we need to find an $(\alpha,c)$-potential function as defined in \cref{defn:potential}. 
The potential $\Psi$ we choose is exactly the one from \cite{LLY13}, adapted to the log marginal ratios and the tree recursion $H$ (see \cref{sec:contraction-potential} for more details). 
We show that if the parameters $(\beta,\gamma,\lambda)$ are up-to-$\Delta$ unique with gap $\delta \in (0,1)$ and either $\sqrt{\beta\gamma} > \frac{\Delta-2}{\Delta}$ or $\gamma \le 1$, then $\Psi$ is an $(\alpha,c)$-potential. 

\begin{lem}
\label{lem:potential-is-good}
Let $\Delta \ge 3$ be an integer. 
Let $\beta,\gamma,\lambda$ be reals such that $0 \le \beta \le \gamma$, $\gamma>0$, $\beta\gamma < 1$ and $\lambda > 0$. 
Assume that $(\beta,\gamma,\lambda)$ is up-to-$\Delta$ unique with gap $\delta \in (0,1)$. 
Define the function $\Psi$ implicitly by
\begin{equation}
\label{eq:Psi}
    \Psi'(y) 
    = \psi(y) 
    = \sqrt{\frac{(1-\beta\gamma)e^{y}}{(\beta e^{y} + 1)(e^{y} + \gamma)}} 
    = \sqrt{\abs{h(y)}}, 
    \qquad
    \Psi(0) = 0. 
\end{equation}
If $\sqrt{\beta\gamma} > \frac{\Delta-2}{\Delta}$, then $\Psi$ is an $(\alpha,c)$-potential function with $\alpha \geq \delta/2$ and $c \leq 1.5$. 
If $\sqrt{\beta\gamma} \le \frac{\Delta-2}{\Delta}$ and $\gamma \le 1$, then $\Psi$ is an $(\alpha,c)$-potential with $\alpha \ge \delta/2$ and $c \leq 18$; we can further take $c \leq 4$ if $\beta = 0$. 
\end{lem}

We deduce \cref{thm:main-2-spin} for the case $\sqrt{\beta\gamma} > \frac{\Delta-2}{\Delta}$ or $\gamma \le 1$ from \cref{thm:contraction-implies-mixing} and \cref{lem:potential-is-good}. 
The proof of it can be found in \cref{sec:proof-main}. 
The case that $\sqrt{\beta\gamma} \le \frac{\Delta-2}{\Delta}$ and $\gamma > 1$ is trickier. 
As discussed in Section 5 of \cite{LLY13}, 
when $\sqrt{\beta\gamma} \le \frac{\Delta-2}{\Delta}$ and $\gamma > 1$, 
for some $\lambda>0$ the spin system lies in the uniqueness region for arbitrary graphs, even with unbounded degrees (i.e., up-to-$\infty$ unique). 
Thus, in this case the total influences of a vertex can be as large as $\Theta(\Delta/\delta)$, resulting in $n^{\Theta(\Delta/\delta)}$ mixing time. 
To deal with this, we consider a suitably weighted sum of absolute influences of a fixed vertex, which also upper bounds the maximum eigenvalue of the influence matrix. 
\cref{defn:potential} and \cref{thm:contraction-implies-mixing} are then modified to a slightly stronger version. 
The statements and proofs for this case are presented in \cref{sec:antiferroweighted} and \cref{sec:weighted-contraction-mixing}.

The rest of the paper is organized as follows. 
In \cref{sec:graph-tree} we prove \cref{lem:Tsaw-G} about properties of the SAW tree. 
In \cref{sec:tree-influence} we establish \cref{lem:bound-for-tree} regarding the decay of influences by the potential method. 
We verify the \hyperlink{cond:contraction}{Contraction} condition in \cref{sec:contraction-potential} for our choice of potential. 
\cref{sec:antiferroweighted} is devoted to the case that $\sqrt{\beta\gamma} \le \frac{\Delta-2}{\Delta}$ and $\gamma > 1$, where a more general version of \cref{defn:potential} and \cref{thm:contraction-implies-mixing} is required; missing proofs can be found in \cref{sec:weighted-contraction-mixing}. 
In \cref{sec:boundedness} we verify the \hyperlink{cond:boundedness}{Boundedness} condition and its generalization for our potential in all cases. 
We consider ferromagnetic spin systems in \cref{sec:ferro} and the proofs are left to \cref{sec:ferroproofs}. 
We prove all of our main results in \cref{sec:proof-main}.

\section{Preservation of influences for self-avoiding walk trees}
\label{sec:graph-tree}
In this section we show that the self-avoiding walk (SAW) tree, introduced in \cite{Wei06} (see also \cite{SS05}),  maintains all the influence of the root, and thus establishes \cref{lem:Tsaw-G}. 
To do this, we show that the partition function of $G$, viewed as a polynomial of the external fields $\mybf{\lambda}$, divides that of the SAW tree. From there we prove that the influence of the root vertex $r$ on another vertex $v$ in $G$, is exactly equal to that on all copies of $v$ in the SAW tree. 
Using our proof approach, we show that the marginal of the root is maintained in the SAW tree, re-establishing Weitz's celebrated result \cite{Wei06}, and also all pairwise covariances concerned with $v$ are preserved. 

\begin{thm}\label{thm:Tsaw-G}
Let $G=(V,E)$ be a connected graph, $r\in V$ be a vertex and $\Lambda \subseteq V \backslash \{r\}$ such that $G\backslash \Lambda$ is connected. 
Let $T = T_{\textsc{saw}}(G,r)$ be the self-avoiding walk tree of $G$ rooted at $r$. 
Then for every $\sL \in \{0,1\}^\Lambda$, $Z^\sL_G$ divides $Z^\sL_T$. 
More precisely, there exists a polynomial $P^\sL_{G,r} = P^\sL_{G,r}(\mybf{\lambda})$ such that
\[
Z^\sL_T = Z^\sL_G \cdot P^\sL_{G,r}. 
\]
Moreover, the polynomial $P^\sL_{G,r}$ is independent of $\lambda_r$. 
\end{thm}
\begin{remark}\label{rem:divisibilityproof}
The proof of \cref{thm:Tsaw-G} can be adapted to give a purely combinatorial proof of \cref{eq:TsawGinf} in \cref{lem:Tsaw-G}. Like in the proof of \cite[Theorem 3.1]{Wei06}, one can proceed via vertex splitting and telescoping, where instead of telescoping a product of marginal ratios, one instead telescopes a sum of single-vertex influences.
\end{remark}

We remark that \cite{Ben18} proved a univariate version of \cref{thm:Tsaw-G} for the hardcore model, and \cite{LSS19} showed a similar result for the zero-field Ising model with a uniform edge weight. 
Our result holds for all $2$-spin systems and arbitrary fields for each vertex. 
We can also generalize it to arbitrary edge weights for each edge in a straightforward fashion. 
It is crucial that the quotient polynomial $P^\sL_{G,r}$ is independent of the field $\lambda_r$ at the root, from which we can deduce the preservation of marginal and influences of the root immediately.  


Before proving \cref{thm:Tsaw-G}, we first give a few consequences of it. 
For all $u,v\in V \backslash \Lambda$, we define the \emph{marginal} at $v$ as $\marginal_G^\sL(v) = \mu_G(v \seq 1 \mid \sL)$ (henceforth we write $v = i$ for the event $\sigma_v = i$ for convenience), 
and the \emph{covariance} of $u$ and $v$ as
\[
\CovM_G^\sL(u,v) = \mu_G(u \seq v \seq 1 \mid \sL) - \mu_G(u \seq 1 \mid \sL) \mu_G(v \seq 1 \mid \sL). 
\]
The following lemma relates the quantities we are interested in with appropriate derivatives of the (log) partition function. Parts 1 and 2 of the lemma are folklore. 

\begin{lem}\label{lem:2spin-sys-property}
For every graph $G=(V,E)$, $\Lambda \subseteq V$ and $\sL \in \{0,1\}^\Lambda$, the following holds:
\begin{enumerate}
\item For all $v \in V$, 
\[
\left(\optheta{v}\right) \log Z^\sL_G = \marginal_G^\sL(v); 
\]
\item For all $u,v \in V$, 
\[
\left(\optheta{v}\right) \left(\optheta{u}\right) \log Z^\sL_G 
= \left(\optheta{v}\right) \marginal_G^\sL(u) 
= \CovM_G^\sL(u,v); 
\]
\item For all $u,v \in V$, 
\[
\left(\optheta{v}\right) \log R_G^\sL(u) = \II_G^\sL(u \sra v). 
\]
\end{enumerate}
\end{lem}

\begin{proof}
The first two parts are standard. 
We include the proofs of these two facts in \cref{sec:proof-2spin} for completeness. 
For Part 3, we deduce from Part 2 that
\begin{align*}
\left(\optheta{v}\right) \log R_G^\sL(u) 
= \left(\optheta{v}\right) \log \left( \frac{\marginal_G^\sL(u)}{1-\marginal_G^\sL(u)} \right) 
= \frac{ \left(\optheta{v}\right) \marginal_G^\sL(u)  }{\marginal_G^\sL(u) \left(1-\marginal_G^\sL(u)\right)} 
= \frac{\CovM_G^\sL(u,v)}{\CovM_G^\sL(u,u)}. 
\end{align*}
It remains to show that 
\[
\II_G^\sL(u \sra v) 
= \frac{\CovM_G^\sL(u,v)}{\CovM_G^\sL(u,u)}, 
\]
which actually holds for any two binary random variables. 
To see this, we first compute $\CovM_G^\sL(u,u) \cdot \II_G^\sL(u \sra v)$ by definition: 
\begin{align*}
&\CovM_G^\sL(u,u) \cdot \II_G^\sL(u \sra v)\\ 
={}& \mu_G(u \seq 1 \mid \sL) \cdot \mu_G(u \seq 0 \mid \sL) \cdot
\left[ \mu_G(v \seq 1 \mid u \seq 1, \,\sL) - \mu_G(v \seq 1 \mid u \seq 0, \,\sL) \right]\\ 
={}& \mu_G(u \seq 1, v \seq 1 \mid \sL) \cdot \mu_G(u \seq 0 \mid \sL) - \mu_G(u \seq 1 \mid \sL) \cdot \mu_G(u \seq 0, v \seq 1 \mid \sL)\\ 
={}&  \mu_G(u \seq 1, v \seq 1 \mid \sL) \cdot \mu_G(u \seq 0, v \seq 0 \mid \sL) - \mu_G(u \seq 1, v \seq 0 \mid \sL) \cdot \mu_G(u \seq 0, v \seq 1 \mid \sL). 
\end{align*}
Meanwhile, the covariance can be written as
\begin{align*}
\CovM_G^\sL(u,v) &= \mu_G(u \seq 1, v \seq 1 \mid \sL) - \mu_G(u \seq 1 \mid \sL) \cdot \mu_G(v \seq 1 \mid \sL)\\ 
&= \mu_G(u \seq 1, v \seq 1 \mid \sL) \cdot \mu_G(u \seq 0, v \seq 0 \mid \sL) - \mu_G(u \seq 1, v \seq 0 \mid \sL) \cdot \mu_G(u \seq 0, v \seq 1 \mid \sL).
\end{align*}
This shows that $\II_G^\sL(u \sra v) = \CovM_G^\sL(u,v) / \CovM_G^\sL(u,u)$ and thus establishes Part 3. 
\end{proof}

We deduce \cref{lem:Tsaw-G} from \cref{thm:Tsaw-G} and the second item of the following lemma. 
The proof of \cref{thm:Tsaw-G} is presented in \cref{sec:proof-SAW-G}. 

\begin{lem}\label{lem:I-G-T}
Let $G=(V,E)$ be a connected graph, $r\in V$ be a vertex and $\Lambda \subseteq V \backslash \{r\}$ such that $G\backslash \Lambda$ is connected. 
Let $T = T_{\textsc{saw}}(G,r)$ be the self-avoiding walk tree of $G$ rooted at $r$. 
Then for every $\sL \in \{0,1\}^\Lambda$ we have:
\begin{enumerate}
\item (\cite[Theorem 3.1]{Wei06}) Preservation of marginal of the root $r$:
\[
\marginal_G^\sL(r) = \marginal_T^\sL(r) \qquad\text{and}\qquad R_G^\sL(r) = R_T^\sL(r);
\]
\item Preservation of covariances and influences of $r$: for every $v\in V$, 
\[
\CovM_G^\sL(r,v) = \sum_{\hat{v} \in \mathcal{C}_v} \CovM_T^\sL(r,\hat{v}) \qquad\text{and}\qquad \II_G^\sL(r\sra v) = \sum_{\hat{v} \in \mathcal{C}_v} \II_T^\sL(r\sra \hat{v}). 
\]
where $\mathcal{C}_v$ is the set of all free (unfixed) copies of $v$ in $T$. 
\end{enumerate}
\end{lem}
\begin{proof}
By \cref{thm:Tsaw-G}, there exists a polynomial $P^\sL_{G,r} = P^\sL_{G,r}(\mybf{\lambda})$ such that $Z^\sL_T = Z^\sL_G \cdot P^\sL_{G,r}$ 
and $P^\sL_{G,r}$ is independent of $\lambda_r$. 
Then it follows from \cref{lem:2spin-sys-property} that
\[
\marginal_T^\sL(r) = \left(\optheta{r}\right) \log Z^\sL_T 
= \left(\optheta{r}\right) \left( \log Z^\sL_G + \log P^\sL_{G,r} \right) 
= \left(\optheta{r}\right) \log Z^\sL_G 
= \marginal_G^\sL(r), 
\]
and therefore $R_T^\sL(r) = R_G^\sL(r)$. 
For the second item, again from \cref{lem:2spin-sys-property} we get
\[
\CovM_G^\sL(r,v) = \left(\optheta{v}\right) \marginal_G^\sL(r) = \left(\optheta{v}\right) \marginal_T^\sL(r). 
\]
Recall that for the spin system on the SAW tree $T$, every free copy $\hat{v}$ of $v$ from $\mathcal{C}_v$ has the same external field $\lambda_{\hat{v}} = \lambda_v$. 
Then, by the chain rule of derivatives and \cref{lem:2spin-sys-property}, we deduce that
\[
\CovM_G^\sL(r,v) = \sum_{\hat{v} \in \mathcal{C}_v} \left(\optheta{\hat{v}}\right) \marginal_T^\sL(r) 
\cdot \frac{\partial \lambda_{\hat{v}}}{\partial \lambda_v} \cdot \frac{\lambda_v}{\lambda_{\hat{v}}} 
= \sum_{\hat{v} \in \mathcal{C}_v} \CovM_T^\sL(r,\hat{v}). 
\]
Finally, we have
\[
\II_G^\sL(r \sra v) = \left(\optheta{v}\right) \log R_G^\sL(r) = \left(\optheta{v}\right) \log R_T^\sL(r) = \sum_{\hat{v} \in \mathcal{C}_v} \II_T^\sL(r\sra \hat{v}), 
\]
where the last equality follows as above. 
\end{proof}

\subsection{Proof of \texorpdfstring{\cref{thm:Tsaw-G}}{Theorem 11}}
\label{sec:proof-SAW-G}


Before presenting our proof, let us first review the notations and definitions introduced earlier. 
Denote the set of fields at all vertices by $\mybf{\lambda} = \{\lambda_v: v\in V\}$. 
For $\Lambda \subseteq V$ and $\sL \in \{0,1\}^\Lambda$, the weight of $\sigma \in \{0,1\}^{V\backslash \Lambda}$ conditional on $\sL$ is given by
\[
w_G(\sigma \mid \sL) = \beta^{m_1(\sigma \mid \sL)} \gamma^{m_0(\sigma \mid \sL)} \prod_{v\in V\backslash \Lambda} \lambda_v^{\sigma_v},
\]
where for $i=0,1$, $m_i(\cdot \mid \sL)$ denotes the number of edges such that both endpoints receive the spin $i$ and at least one of them is in $V\backslash \Lambda$. 
The partition function conditional on $\sL$ is defined as $Z_G^\sL = \sum_{\sigma \in \{0,1\}^{V\backslash \Lambda}} w_G(\sigma \mid \sL)$. 
For the SAW tree, we define the conditional weights and partition function in the same way. 
In particular, recall that when we fix a conditioning $\sL$ on the SAW tree, we also remove all descendants of $\hat{v} \in \mathcal{C}_v$ for each $v\in \Lambda$. 

For every $v\in V\backslash \Lambda$ and $i\in \{0,1\}$, we shall write $v = i$ to represent the set of configurations such that $\sigma_v = i$ (i.e., $\{\sigma \in \{0,1\}^{V\backslash \Lambda}: \sigma_v = i\}$) and let $Z_G^\sL(v \seq i)$ be sum of weights of all configurations with $v \seq i$. 
We further extend this notation and write $Z_G^\sL(U \seq \sigma_U)$ for every $U\subseteq V\backslash \Lambda$ and $\sigma_U \in \{0,1\}^U$. 
For the SAW tree we adopt the same notations as well. 

\begin{proof}[Proof of \cref{thm:Tsaw-G}]
We will show that there exists a polynomial $P^\sL_{G,r} = P^\sL_{G,r}(\mybf{\lambda})$, independent of $\lambda_r$, such that
\begin{equation}
\label{eq:SAW-induction}
Z^\sL_T(r \seq 1) = Z^\sL_G(r \seq 1) \cdot P^\sL_{G,r} 
\quad\text{and}\quad
Z^\sL_T(r \seq 0) = Z^\sL_G(r \seq 0) \cdot P^\sL_{G,r}.
\end{equation}
The high-level proof idea of \cref{eq:SAW-induction} is similar to the corresponding result in \cite[Theorem 3.1]{Wei06}. 
Let $m$ be the number of edges with at least one endpoint in $V\backslash \Lambda$. We use induction on $m$. 
When $m = 0$ the statement is trivial since $T=G$. 
Assume that \cref{eq:SAW-induction} holds for all graphs and all conditioning with less than $m$ edges. 
Suppose that the root $r$ has $d$ neighbors $v_1,\dots,v_d$. 
Define $G'$ to be the graph obtained by replacing the vertex $r$ with $d$ vertices $r_1,\dots,r_d$ and then connecting $\{r_i,d_i\}$ for $1\le i\le d$. 

Consider first the case where $(G\backslash \{r\}) \backslash \Lambda$ is still connected. 
For each $i$, let $G_i = G' - r_i$. 
Define the $2$-spin system on $G_i$ with the same parameters $(\beta,\gamma,\mybf{\lambda})$, plus an additional conditioning that the vertices $r_1,\dots, r_{i-1}$ are fixed to spin $0$ while $r_{i+1},\dots, r_d$ are fixed to spin $1$; we denote this conditioning by $\sigma_{U_i}$ with $U_i = \{v_1,\dots,v_d\} \backslash \{v_i\}$. 
Then, $T=\TSAW(G,r)$ can be generated by the following recursive procedure. 
Also see \cref{fig:saw_construction} for an illustration. 

\begin{figure}[t]
\centering
\includegraphics[width = \textwidth]{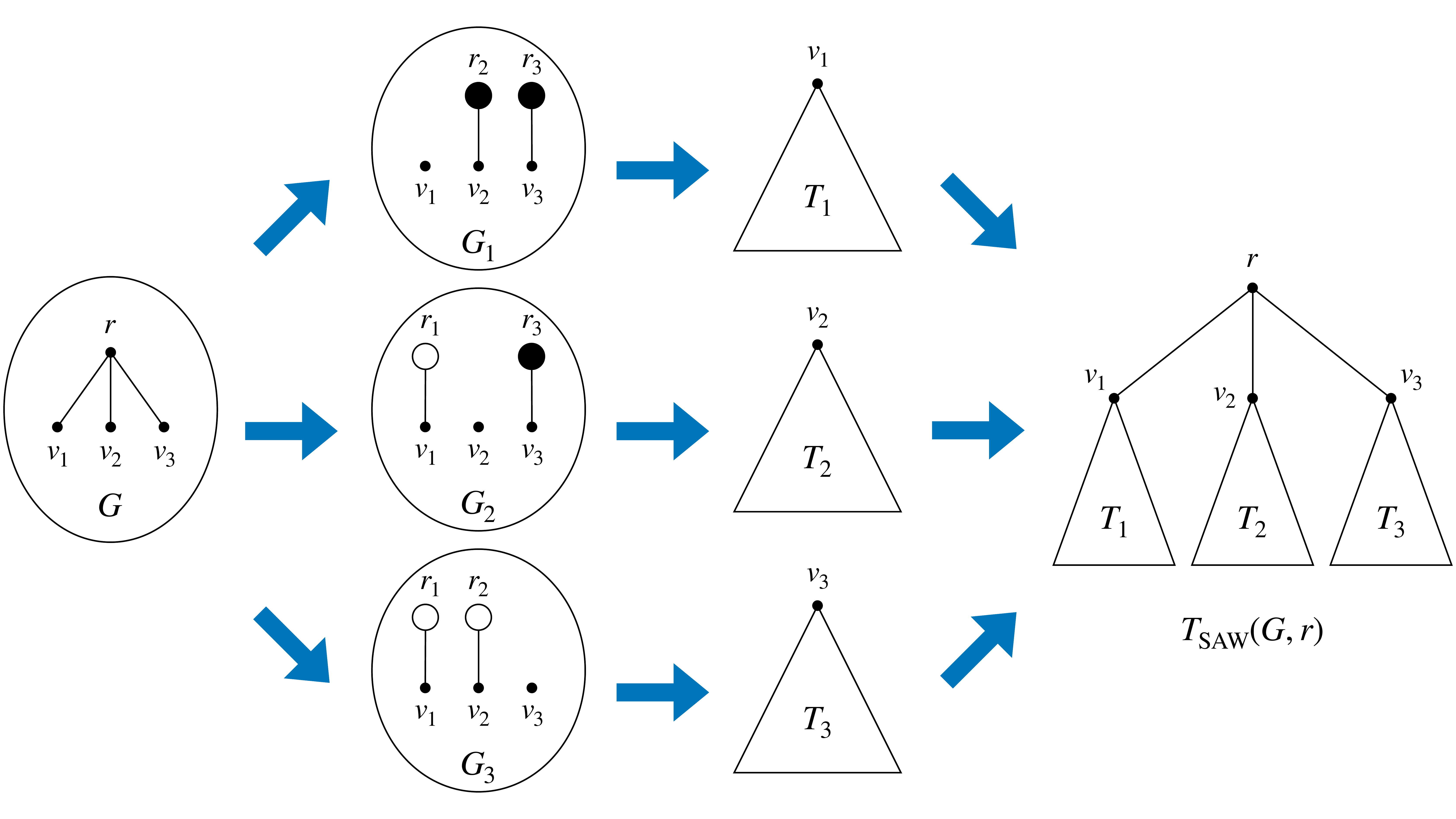}
\caption{A recursive construction of the self-avoiding walk (SAW) tree. Here $T_i$ is the SAW tree of $G_i$ rooted at $v_i$ for $i=1,2,3$. ($\CIRCLE$/$\Circle$: fixed to spin $1$/$0$.)}
\label{fig:saw_construction}
\end{figure}

\medskip
\paragraph{Algorithm: $\TSAW(G,r)$}
\begin{enumerate}
\item For each $i$, 
let $T_i = \TSAW(G_i,v_i)$ plus the conditioning $\sigma_{U_i}$; 
\item Let $T=\TSAW(G,r)$ be the union of $r$ and $T_1,\dots,T_d$ by connecting $\{r,v_i\}$ for $1\le i\le d$; output $T$. 
\end{enumerate}
\medskip

For the purpose of proof, we also consider the $2$-spin system on $G'$ with the same parameters $(\beta,\gamma,\mybf{\lambda})$, with an exception that we let the vertices $r_1,\dots,r_d$ have no fields (i.e., setting $\lambda_{r_i} = 1$ for $1\le i\le d$ instead of $\lambda_r$). 
We then observe that
\[
Z^\sL_G(r \seq 1) = \lambda_r \cdot Z^\sL_{G'}(r_1 \seq 1, \dots, r_d \seq 1),
\]
and the same holds with spin $1$ replaced by $0$. 
For $1\le i \le d$, let $\sigma_{\Lambda_i}$ denote the union of the conditioning $\sL$ and $\sigma_{U_i}$, where $\Lambda_i = \Lambda \cup U_i$. 
Then for every $1\le i\le d$ we have
\[
Z^\sL_{G'}(r_1 \seq 0, \dots, r_{i-1} \seq 0, r_i \seq 1, \dots, r_d \seq 1) 
= \beta \cdot Z^{\sigma_{\Lambda_i}}_{G_i}(v_i \seq 1) + Z^{\sigma_{\Lambda_i}}_{G_i}(v_i \seq 0).
\]
Notice that both sides are independent of the field $\lambda_r$: for the left side, all $r_i$'s do not have a field for the spin system on $G'$; for the right side, recall that we do not count the weight of fixed vertices for the conditional partition function for each $G_i$. 
Now define $Q^\sL_{G,r} = Q^\sL_{G,r}(\mybf{\lambda})$ by 
\[
Q^\sL_{G,r} = \prod_{i=2}^d Z^\sL_{G'}(r_1 \seq 0, \dots, r_{i-1} \seq 0, r_i \seq 1, \dots, r_d \seq 1), 
\]
which is independent of $\lambda_r$. 
Then we get
\begin{align*}
Z^\sL_G(r \seq 1) \cdot Q^\sL_{G,r} &= 
\lambda_r \cdot \prod_{i=1}^d Z^\sL_{G'}(r_1 \seq 0, \dots, r_{i-1} \seq 0, r_i \seq 1, \dots, r_d \seq 1)\\ 
&= \lambda_r \cdot \prod_{i=1}^d 
\left( \beta \cdot Z^{\sigma_{\Lambda_i}}_{G_i}(v_i \seq 1) + Z^{\sigma_{\Lambda_i}}_{G_i}(v_i \seq 0) \right).
\end{align*}
Using a similar argument, we also have
\begin{align*}
Z^\sL_G(r \seq 0) \cdot Q^\sL_{G,r} &= 
\prod_{i=1}^d Z^\sL_{G'}(r_1 \seq 0, \dots, r_i \seq 0, r_{i+1} \seq 1, \dots, r_d \seq 1)\\ 
&= \prod_{i=1}^d 
\left( Z^{\sigma_{\Lambda_i}}_{G_i}(v_i \seq 1) + \gamma \cdot Z^{\sigma_{\Lambda_i}}_{G_i}(v_i \seq 0) \right).
\end{align*}
Since we assume that $(G\backslash \{r\}) \backslash \Lambda$ is connected, the graph $G_i \backslash \Lambda$ is also connected for each $i$. 
Then, by the induction hypothesis, for each $i$ there exists a polynomial $P^{\sigma_{\Lambda_i}}_{G_i,v_i} = P^{\sigma_{\Lambda_i}}_{G_i,v_i}(\mybf{\lambda})$ such that
\[
Z^{\sigma_{\Lambda_i}}_{T_i}(r \seq 1) = Z^{\sigma_{\Lambda_i}}_{G_i}(r \seq 1) \cdot P^{\sigma_{\Lambda_i}}_{G_i,v_i} 
\quad\text{and}\quad
Z^{\sigma_{\Lambda_i}}_{T_i}(r \seq 0) = Z^{\sigma_{\Lambda_i}}_{G_i}(r \seq 0) \cdot P^{\sigma_{\Lambda_i}}_{G_i,v_i}; 
\]
these polynomials are independent of $\lambda_r$ since the conditional partition functions for $G_i$'s do not involve $\lambda_r$. 
Now if we let 
\[
P^\sL_{G,r} = Q^\sL_{G,r} \cdot \prod_{i=1}^d P^{\sigma_{\Lambda_i}}_{G_i,v_i},
\]
then it follows from the tree recursion that
\begin{align*}
Z^\sL_T(r \seq 1) &= 
\lambda_r \cdot \prod_{i=1}^d 
\left( \beta \cdot Z^{\sigma_{\Lambda_i}}_{T_i}(v_i \seq 1) + Z^{\sigma_{\Lambda_i}}_{T_i}(v_i \seq 0) \right)\\ 
&= \lambda_r \cdot \prod_{i=1}^d 
\left( \beta \cdot Z^{\sigma_{\Lambda_i}}_{G_i}(v_i \seq 1) \cdot P^{\sigma_{\Lambda_i}}_{G_i,v_i} + Z^{\sigma_{\Lambda_i}}_{G_i}(v_i \seq 0) \cdot P^{\sigma_{\Lambda_i}}_{G_i,v_i} \right)\\
&= Z^\sL_G(r \seq 1) \cdot Q^\sL_{G,r} \cdot \prod_{i=1}^d P^{\sigma_{\Lambda_i}}_{G_i,v_i}\\ 
&= Z^\sL_G(r \seq 1) \cdot P^\sL_{G,r}. 
\end{align*}
The other equality $Z^\sL_T(r \seq 0) = Z^\sL_G(r \seq 0) \cdot P^\sL_{G,r}$ is established in the same way. This completes the proof for the case that $(G\backslash \{r\}) \backslash \Lambda$ is connected. 

If $(G \backslash \{r\}) \backslash \Lambda$ has two or more connected components, then we can construct $\TSAW(G,r)$ by the SAW tree of each component. 
Recall that $G'$ is defined by splitting the vertex $r$ into $d$ copies in the graph $G$. 
Suppose that $G'\backslash \Lambda$ has $k$ connected component for an integer $k\ge 2$. 
Let $G'_{(1)},\dots, G'_{(k)}$ be the subgraphs induced by each component, along with vertices from $\Lambda$ that are adjacent to it. 
For each $j$, let $G_{(j)}$ be the graph obtained from $G'_{(j)}$ by contracting all copies of $r$ into one vertex $r_{(j)}$, and let $T_{(j)} = \TSAW(G'_{(j)}, r_{(j)})$. 
Observe that once we contract the roots $r_{(1)},\dots, r_{(k)}$ of $T_{(1)},\dots, T_{(k)}$, the resulting tree is $\TSAW(G,r)$. 

We define the $2$-spin system on each $G_{(j)}$ with the same parameters $(\beta,\gamma,\mybf{\lambda})$, except that the vertex $r_{(j)}$ does not have a field (i.e., $\lambda_{r_{(j)}} = 1$ instead of $\lambda_r$). 
For $1\le j \le k$, let $\Lambda_{(j)} = \Lambda \cap V(G_{(j)})$ and $\sigma_{\Lambda_{(j)}}$ be the configuration $\sL$ restricted on $\Lambda_{(j)}$. 
Then $G_{(j)} \backslash \Lambda_{(j)}$ is connected for every $j$ and, 
since $k\ge 2$, each $G_{(j)}$ with conditioning $\sigma_{\Lambda_{(j)}}$ has fewer than $m$ edges. 
Thus, we can apply the induction hypothesis; 
namely, for $1\le j \le k$ there exists a polynomial $P^{\sigma_{\Lambda_{(j)}}}_{G_{(i)},r_{(i)}} = P^{\sigma_{\Lambda_{(j)}}}_{G_{(i)},r_{(i)}}(\mybf{\lambda})$, which is independent of $\lambda_r$, such that
\[
Z^{\sigma_{\Lambda_{(j)}}}_{T_{(j)}}(r_{(j)} \seq 1) = Z^{\sigma_{\Lambda_{(j)}}}_{G_{(j)}}(r_{(j)} \seq 1) \cdot P^{\sigma_{\Lambda_{(j)}}}_{G_{(j)},r_{(j)}} 
\quad\text{and}\quad
Z^{\sigma_{\Lambda_{(j)}}}_{T_{(j)}}(r_{(j)} \seq 0) = Z^{\sigma_{\Lambda_{(j)}}}_{G_{(j)}}(r_{(j)} \seq 0) \cdot P^{\sigma_{\Lambda_{(j)}}}_{G_{(j)},r_{(j)}}. 
\]
We define the polynomial $P^\sL_{G,r} = P^\sL_{G,r}(\mybf{\lambda})$ to be
\[
P^\sL_{G,r} = \prod_{j=1}^k P^{\sigma_{\Lambda_{(j)}}}_{G_{(j)}, r_{(j)}}. 
\]
It is then easy to check that
\begin{align*}
Z^\sL_T(r \seq 1) 
&= \lambda_r \cdot \prod_{j=1}^k Z^{\sigma_{\Lambda_{(j)}}}_{T_{(j)}}(r_{(j)} \seq 1) 
= \lambda_r \cdot \prod_{j=1}^k \left( Z^{\sigma_{\Lambda_{(j)}}}_{G_{(j)}}(r_{(j)} \seq 1) \cdot P^{\sigma_{\Lambda_{(j)}}}_{G_{(j)}, r_{(j)}} \right)\\ 
&= Z^\sL_G(r \seq 1) \cdot \prod_{j=1}^k P^{\sigma_{\Lambda_{(j)}}}_{G_{(j)}, r_{(j)}} 
= Z^\sL_G(r \seq 1) \cdot P^\sL_{G,r},  
\end{align*}
and similarly $Z^\sL_T(r \seq 0) = Z^\sL_G(r \seq 0) \cdot P^\sL_{G,r}$. The theorem then follows. 
\end{proof}

\section{Influence bound for trees}
\label{sec:tree-influence}
In this section, we study the influences of the root on other vertices in a tree. 
We give an upper bound on the total influences of the root on all vertices at a fixed distance away. 
To do this, we apply the potential method, which has been used to establish the correlation decay property (see, e.g., \cite{LLY12,LLY13,GL18}). 
Given an arbitrary potential function $\Psi$, our upper bound is in terms of properties of $\Psi$, involving bounds on $\norm{\grad H_d^\Psi}_1$ and $|\psi|$ where $\psi = \Psi'$. 
We then deduce \cref{lem:bound-for-tree} in the case that $\Psi$ an $(\alpha,c)$-potential.

Assume that $T=(V_T,E_T)$ is a tree rooted at $r$ of maximum degree at most $\Delta$. 
Let $\Lambda \subseteq V_T \backslash \{r\}$ and $\sL \in \{0,1\}^\Lambda$ be arbitrary and fixed. 
Consider the $2$-spin system on $T$ with parameters $(\beta,\gamma,\lambda)$, conditioned on $\sL$. 
We need to bound the influence $\II_T^\sL(r\sra v)$ from the root $r$ to another vertex $v\in V_T$. 
Notice that if $v$ is disconnected from $r$ when $\Lambda$ is removed, then $\II_T^\sL(r\sra v) = 0$ by the Markov property of spin systems. 
Therefore, we may assume that, by removing all such vertices, $\Lambda$ contains only leaves of $T$. 

For a vertex $v\in V_T$, let $T_v = (V_{T_v}, E_{T_v})$ be the subtree of $T$ rooted at $v$ that contains all descendant of $v$; note that $T_r = T$. 
We will write $L_{v}(k) \subseteq V_T \backslash \Lambda$ for the set of all free vertices at distance $k$ away from $v$ in $T_v$. 
We pay particular interest in the marginal ratio at $v$ in the subtree $T_v$, and write $R_v = R_{T_v}^\sL(v)$ for simplicity. 
The $\log R_v$'s are related by the tree recursion $H$. 
If a vertex $v$ has $d$ children, denoted by $u_1,\dots,u_d$, then the tree recursion is given by 
\[
\log R_v = H_d(\log R_{u_1}, \dots, \log R_{u_d}), 
\]
where for $1\le d \le \Delta$ and $(y_1,\dots,y_d) \in[-\infty,+\infty]^d$, 
\[
H_d(y_1,\dots,y_d) = \log \lambda + \sum_{i=1}^d \log \left( \frac{\beta e^{y_i} + 1}{e^{y_i} + \gamma} \right). 
\]
Also recall that for $y\in[-\infty,+\infty]$, we define 
\[
h(y) = - \frac{(1-\beta\gamma) e^y}{(\beta e^y + 1)(e^y + \gamma)} 
\]
and $\frac{\partial}{\partial y_i} H_d(y_1,\dots,y_d) = h(y_i)$ for all $1\le i\le d \le \Delta$. 

The following lemma allows us to bound the sum of all influences from the root to distance $k$, using an arbitrary potential function. 
\begin{lem}
\label{lem:A-and-B}
Let $\Psi: [-\infty,+\infty] \to (-\infty,+\infty)$ be a differentiable and increasing (potential) function with image $S = \Psi[-\infty,+\infty]$ and derivative $\psi = \Psi'$. 
Denote the degree of the root $r$ by $\Delta_r$. 
Then for every integer $k \ge 1$, 
\begin{equation*}
    \sum_{v \in L_r(k)} \left| \mathcal{I}_{T}^\sL(r\sra v) \right|
    \le \Delta_r A_\Psi B_\Psi 
    \left( \max_{1\le d < \Delta} \,\sup_{\tilde{\mybf{y}} \in S^d} \norm{\grad H_{d}^{\Psi} (\tilde{\mybf{y}})}_1 \right)^{k-1}
\end{equation*}
where 
\[
A_\Psi = \max_{u\in L_r(1)} \left\{ \frac{|h(\log R_u)|}{\psi(\log R_u)} \right\} 
\quad\text{and}\quad 
B_\Psi = \max_{v\in L_r(k)} \left\{ \psi(\log R_v) \right\}. 
\]
\end{lem}


Before proving \cref{lem:A-and-B}, we first present two useful properties of the influences on trees. 
Firstly, it was shown in \cite{ALO20} that the influences satisfy the following form of chain rule on trees. 

\begin{lem}[{\cite[Lemma B.2]{ALO20}}]
\label{lem:inf-chain-rule}
Suppose that $u,v,w \in V_T$ are three distinct vertices such that $u$ is on the unique path from $v$ to $w$. Then 
\[
\II_T^\sL(v \sra w) = \II_T^\sL(v \sra u) \cdot \II_T^\sL(u \sra w). 
\]
\end{lem}

Secondly, for two adjacent vertices on a tree, the influence from one to the other is given by the function $h$.  

\begin{lem}
\label{lem:inf-h}
Let $v \in V_T$ and $u$ be a child of $v$ in the subtree $T_v$. Then
\[
\II_T^\sL(v \sra u) = h(\log R_u).
\]
\end{lem}

\begin{proof}
The lemma can be proved through an explicit computation of the influence. 
Here we present a more delicate proof utilizing \cref{lem:2spin-sys-property}, which gives some insights into the relation between the influence and the function $h$. 
We assume that $v$ has $d$ children in the subtree $T_v$, denoted by $u_1=u$ and $u_2, \dots, u_d$ respectively. 
We also assume, as a more general setting than uniform fields, that each vertex $w$ is attached to a field $\lambda_w$ of its own. 
Then \cref{lem:2spin-sys-property} and the tree recursion imply that
\begin{align*}
\II_T^\sL(v \sra u) 
&= \II_{T_v}^\sL(v \sra u)
= \left(\optheta{u}\right) \log R_v \\
&= \left(\optheta{u}\right) H_d(\log R_{u_1}, \dots, \log R_{u_d}) \\
&= \sum_{i=1}^d \frac{\partial}{\partial \log R_{u_i}} H_d(\log R_{u_1}, \dots, \log R_{u_d}) \cdot \left(\optheta{u}\right) \log R_{u_i} \\
&= \sum_{i=1}^d h(\log R_{u_i}) \cdot \II_{T_{u_i}}^\sL (u_i \sra u)
= h(\log R_u),
\end{align*}
where the last equality is because $\II_{T_{u_i}}^\sL (u_i \sra u) = 0$ for $u_i \neq u$ and $\II_{T_{u}}^\sL (u \sra u) = 1$. 
\end{proof}

We are now ready to prove \cref{lem:A-and-B}. 

\begin{proof}[Proof of \cref{lem:A-and-B}]
For a vertex $v\in V_T$, denote the number of its children by $d_v$; note that $d_r = \Delta_r$. 
Let $u_1,\dots,u_{\Delta_r}$ be the children of the root $r$. 
We may assume that all these children of $r$ are free, since if $u_i$ is fixed then $\II_T^\sL(r \sra u_i) = 0$ by definition. 
Then by \cref{lem:inf-chain-rule} and \cref{lem:inf-h}, we get 
\begin{align*}
 	\sum_{v \in L_r(k)} \left| \II_T^\sL(r \sra v) \right| 
    &= \sum_{i=1}^{\Delta_r} \left| \II_T^\sL(r \sra u_i) \right| \sum_{v\in L_{u_i}(k-1)} \left| \II_T^\sL(u_i \sra v) \right|\\ 
    &= \sum_{i=1}^{\Delta_r} \left| h(\log R_{u_i}) \right| \sum_{v\in L_{u_i}(k-1)} \left| \II_T^\sL(u_i \sra v) \right|\\
    &= \sum_{i=1}^{\Delta_r} \frac{\left| h(\log R_{u_i}) \right|}{\psi(\log R_{u_i})} 
    \sum_{v\in L_{u_i}(k-1)} \psi(\log R_{u_i}) \left| \II_T^\sL(u_i \sra v) \right|.
\end{align*}
Hence, we obtain that
\begin{equation}\label{eq:rt-influence}
    \sum_{v \in L_r(k)} \left| \II_T^\sL(r \sra v) \right| 
    \le \Delta_r \cdot 
    \max_{1\le i \le \Delta_r} \left\{ \frac{|h(\log R_{u_i})|}{\psi(\log R_{u_i})} \right\} \cdot 
    \max_{1\le i \le \Delta_r} \left\{ \sum_{v\in L_{u_i}(k-1)} \psi(\log R_{u_i}) \left| \II_T^\sL(u_i \sra v) \right| \right\}.
\end{equation}
    
Next, we show by induction that for every vertex $u\in V_T \backslash \{r\}$ and every integer $k\ge 0$ we have
\begin{equation}\label{eq:induct}
    \sum_{v\in L_u(k)} 
    \psi(\log R_u) \left| \II_T^\sL(u \sra v) \right|
    \le 
    \max_{v\in L_u(k)} \left\{ \psi(\log R_v) \right\} \cdot 
    \left( \max_{w\in V_{T_u}} \sup_{\tilde{\mybf{y}} \in S^{d_w}} \norm{\grad H_{d_w}^{\Psi} ( \tilde{\mybf{y}} )}_1 \right)^k.  
\end{equation}
Observe that once we establish \cref{eq:induct}, the lemma follows immediately by plugging \cref{eq:induct} into \cref{eq:rt-influence}. 
We will use induction on $k$ to prove \cref{eq:induct}. 
When $k=0$, if $u \in \Lambda$ is fixed then $L_u(0) = \emptyset$ and there is nothing to show; otherwise, \cref{eq:induct} becomes
\[
\psi(\log R_u) \left| \II_T^\sL(u \sra u) \right| \le \psi(\log R_u),
\]
which holds with equality since $\II_T^\sL(u \sra u) = 1$.  
Now suppose that \cref{eq:induct} holds for some integer $k-1 \ge 0$ (and for every vertex $u\in V_T\backslash\{r\}$). 
Let $u\in V_T\backslash\{r\}$ be arbitrary and denote the children of $u$ by $w_1,\dots,w_{d}$, where $1\le d<\Delta$ (if $d=0$ then $L_u(k) = \emptyset$ and \cref{eq:induct} holds trivially). 
Again by \cref{lem:inf-chain-rule} and \cref{lem:inf-h} we have
\begin{align*}
\sum_{v\in L_u(k)} \psi(\log R_u) \left| \II_T^\sL(u \sra v) \right| 
&= \sum_{i=1}^{d} \psi(\log R_u) \left| \II_T^\sL(u \sra w_i) \right| 
\sum_{v \in L_{w_i}(k-1)} \left| \II_T^\sL(w_i \sra v) \right|\\ 
&= \sum_{i=1}^{d} \frac{\psi(\log R_u)}{\psi(\log R_{w_i})} \left| h(\log R_{w_i}) \right|
\sum_{v \in L_{w_i}(k-1)} \psi(\log R_{w_i}) \left| \II_T^\sL(w_i \sra v) \right|. 
\end{align*}
Using the induction hypothesis, we get
\begin{align*}
&\sum_{v\in L_u(k)} \psi(\log R_u) \left| \II_T^\sL(u \sra v) \right| \\ 
\le{}& 
\sum_{i=1}^{d} \frac{\psi(\log R_u)}{\psi(\log R_{w_i})} \left| h(\log R_{w_i}) \right|
\cdot \max_{v\in L_{w_i}(k-1)} \left\{ \psi(\log R_v) \right\} 
\cdot \left( \max_{w\in V_{T_{w_i}}} \sup_{\tilde{\mybf{y}} \in S^{d_w}} \norm{\grad H_{d_w}^{\Psi} (\tilde{\mybf{y}})}_1 \right)^{k-1}\\ 
\le{}& 
\max_{v\in L_u(k)} \left\{ \psi(\log R_v) \right\}
\cdot \left( \max_{w\in V_{T_u} \backslash \{u\}} \sup_{\tilde{\mybf{y}} \in S^{d_w}} \norm{\grad H_{d_w}^{\Psi} (\tilde{\mybf{y}})}_1 \right)^{k-1} 
\cdot \sum_{i=1}^{d} \frac{\psi(\log R_u)}{\psi(\log R_{w_i})} \left| h(\log R_{w_i}) \right|\\ 
\le{}& 
\max_{v\in L_u(k)} \left\{ \psi(\log R_v) \right\} 
\cdot \left( \max_{w\in V_{T_u}} \sup_{\tilde{\mybf{y}} \in S^{d_w}} \norm{\grad H_{d_w}^{\Psi} (\tilde{\mybf{y}})}_1 \right)^k, 
\end{align*}
where the last inequality follows from that
\begin{align*}
\sum_{i=1}^{d} \frac{\psi(\log R_u)}{\psi(\log R_{w_i})} \left| h(\log R_{w_i}) \right| 
&= \sum_{i=1}^{d} \left| \frac{\partial}{\partial \Psi(\log R_{w_i})} \, H_{d}^\Psi \left( \Psi(\log R_{w_1}), \dots, \Psi(\log R_{w_{d}}) \right) \right|\\
&= \norm{\grad H_{d}^\Psi \left( \Psi(\log R_{w_1}), \dots, \Psi(\log R_{w_{d}}) \right)}_1.
\end{align*}
This establishes \cref{eq:induct}, and thus completes the proof of the lemma.     
\end{proof}

We then derive \cref{lem:bound-for-tree} as a corollary. 
\begin{proof}[Proof of \cref{lem:bound-for-tree}]
Since $\Psi$ is an $(\alpha,c)$-potential, the \hyperlink{cond:contraction}{Contraction} condition implies that 
\[
\max_{1\le d < \Delta} \sup_{\tilde{\mybf{y}} \in S^d} \norm{\grad H_{d}^{\Psi} (\tilde{\mybf{y}})}_1 \le 1-\alpha. 
\]
Meanwhile, since the degree of a vertex $v\in V_T \backslash\{r\}$ in the subtree $T_v$ is less than $\Delta$,  we have $\log R_v \in J$. 
Then the \hyperlink{cond:boundedness}{Boundedness} condition implies that for all $u\in L_r(1)$ and $v\in L_r(k)$, 
\[
\frac{\psi(\log R_v)}{\psi(\log R_u)} \cdot |h(\log R_u)| \le \frac{c}{\Delta}. 
\]
Therefore, we get
\[
\Delta_{r} A_\Psi B_\Psi = \Delta_r \cdot \max_{u\in L_r(1)} \left\{ \frac{|h(\log R_u)|}{\psi(\log R_u)} \right\} \cdot \max_{v\in L_r(k)} \left\{ \psi(\log R_v) \right\} \le c. 
\]
The lemma then follows immediately from \cref{lem:A-and-B}. 
\end{proof}

\section{Verifying a good potential: Contraction}
\label{sec:contraction-potential}
In this section, we make a first step for proving \cref{lem:potential-is-good}. 
Let $\Delta \ge 3$ be an integer. 
Let $\beta,\gamma,\lambda$ be reals such that $0\le \beta \le \gamma$, $\gamma>0$, $\beta\gamma <1$ and $\lambda>0$. 
Recall that define our potential function $\Psi: [-\infty,+\infty] \to (-\infty,+\infty)$ through its derivative by 
\begin{equation*}
    \Psi'(y) 
    = \psi(y) 
    = \sqrt{\frac{(1-\beta\gamma)e^{y}}{(\beta e^{y} + 1)(e^{y} + \gamma)}}, 
    \qquad
    \Psi(0) = 0. \tag{1}
\end{equation*}
We include a short proof in \cref{sec:well-def-potential} to show that $\Psi$ is well-defined. 
If $(\beta,\gamma,\lambda)$ is up-to-$\Delta$ unique with gap $\delta \in (0,1)$, 
then we show that $\Psi$ satisfies the \hyperlink{cond:contraction}{Contraction} condition for $\alpha = \delta/2$. 
This holds for all parameters $(\beta,\gamma,\lambda)$ in the uniqueness region, \emph{without} requiring that $\gamma \le 1$. 
Later in \cref{sec:boundedness}, we establish the \hyperlink{cond:boundedness}{Boundedness} condition for $\Psi$ when $\gamma \le 1$, 
completing the proof of \cref{lem:potential-is-good}. 
The case of $\gamma > 1$ is more complicated and is left to \cref{sec:antiferroweighted}.

Before giving our proof, we first point out that the potential function $\Psi$ is essentially the same potential function $\Phi$ used in \cite{LLY13} (notice that \cite{LLY13} uses $\varphi$ as the notation of the potential function and $\Phi = \varphi'$ for its derivative). 
Recall that the tree recursion for the marginal ratios is given by the function $F_d: [0,+\infty]^d \to [0,+\infty]$ where $1\le d\le \Delta$ such that for all $(x_1,\dots,x_d) \in [0,+\infty]^d$, 
\[
F_d(x_1,\dots,x_d) = \lambda \prod_{i=1}^d \frac{\beta x_i + 1}{x_i + \gamma}. 
\]
The potential function $\Phi: [0,+\infty] \to (-\infty,+\infty)$ from \cite{LLY13} is defined implicitly via its derivative as
\begin{equation*}
    \Phi'(x) = \varphi(x) = \frac{1}{\sqrt{x(\beta x + 1)(x + \gamma)}}, \qquad \Phi(1) = 0.
\end{equation*}
The follows lemma explains how we obtain our potential $\Psi$ from $\Phi$. 
\begin{lem}
\label{lem:Psi-Phi}
We have $ \Psi = \sqrt{1-\beta \gamma} \cdot (\Phi \circ \exp) $; 
namely, 
$ \Psi(y) = \sqrt{1-\beta \gamma} \cdot \Phi(e^y) $ 
for all $y\in[-\infty,+\infty]$. 
\end{lem}
\begin{proof}
It is straightforward to check that
\[
\psi(y) = \sqrt{\frac{(1-\beta\gamma)e^{y}}{(\beta e^{y} + 1)(e^{y} + \gamma)}} = \sqrt{1-\beta \gamma} \cdot e^y \cdot \sqrt{\frac{1}{e^y(\beta e^{y} + 1)(e^{y} + \gamma)}} = \sqrt{1-\beta \gamma} \cdot e^y \varphi(e^y). 
\]
Therefore, 
\[
\Psi(y) = \int_{0}^{y} \psi(t) \,\mathrm{d}t 
= \sqrt{1-\beta \gamma} \cdot \int_{0}^{y} e^t \varphi(e^t) \,\mathrm{d}t 
= \sqrt{1-\beta \gamma} \cdot \int_{1}^{e^y} \varphi(s) \,\mathrm{d}s 
= \sqrt{1-\beta \gamma} \cdot \Phi(e^y). \qedhere
\]
\end{proof}

Combining the results of Lemmas 12, 13 and 14 from \cite{LLY13}, we get that the potential function $\Phi$ satisfies the following gradient bound when $(\beta,\gamma,\lambda)$ is in the uniqueness region. 
Note that this can be regarded as the \hyperlink{cond:contraction}{Contraction} condition but for $\Phi$ and $F_d$. 
\begin{thm}[\cite{LLY13}]
\label{thm:potentialgradientbound}
Let $S_\Phi = \Phi[0,+\infty]$ be the image of $\Phi$. 
If the parameters $(\beta,\gamma,\lambda)$ are up-to-$\Delta$ unique with gap $\delta \in (0,1)$, then for every integer $d$ such that $1\le d < \Delta$ and every $(\tilde{x}_1,\dots,\tilde{x}_d) \in S_\Phi^d$, 
\[
\norm{\grad F_d^\Phi(\tilde{x}_1,\dots,\tilde{x}_d)}_1 \le \sqrt{1-\delta}
\]
where $F_d^\Phi = \Phi \circ F_d \circ \Phi^{-1}$. 
\end{thm}

Recall our definition from \cref{sec:mixing-potential}. 
The tree recursion, in terms of the log marginal ratios, is described by the function $H_d: [-\infty, +\infty]^d \to [-\infty, +\infty]$ where $1\le d \le \Delta$ such that for every $(y_1,\dots,y_d)\in [-\infty, +\infty]^d$, 
\[
H_d(y_1,\dots,y_d) = \log \lambda + \sum_{i=1}^d \log \left( \frac{\beta e^{y_i} + 1}{e^{y_i} + \gamma} \right). 
\]
Observe that $H_d = \log \circ F_d \circ \exp$, since we move from ratios to log ratios. 
We are now ready to establish the \hyperlink{cond:contraction}{Contraction} condition for $\Psi$. 

\begin{lem}
\label{lem:contraction-Psi}
Let $S_\Psi = \Psi[-\infty,+\infty]$ be the image of $\Psi$. 
If the parameters $(\beta,\gamma,\lambda)$ are up-to-$\Delta$ unique with gap $\delta \in (0,1)$, then for every integer $d$ such that $1\le d < \Delta$ and every $(\tilde{y}_1,\dots,\tilde{y}_d) \in S_\Psi^d$, 
\[
\norm{\grad H_d^\Psi(\tilde{y}_1,\dots,\tilde{y}_d)}_1 \le \sqrt{1-\delta}
\]
where $H_d^\Psi = \Psi \circ H_d \circ \Psi^{-1}$. 
\end{lem}

\begin{proof}
Define the linear function $a: \R \to \R$ to be $a(x) = \sqrt{1-\beta\gamma} \cdot x$ for $x\in \R$. 
Then \cref{lem:Psi-Phi} gives $\Psi = a \circ \Phi \circ \exp$, and thereby $\Psi \circ \log = a \circ \Phi$. 
It follows that for every $1\le d < \Delta$,
\[
H_d^\Psi = \Psi \circ H_d \circ \Psi^{-1} 
= \Psi \circ \log \circ F_d \circ \exp \circ \Psi^{-1} 
= a \circ \Phi \circ F_d \circ \Phi^{-1} \circ a^{-1} 
= a \circ F_d^\Phi \circ a^{-1}. 
\]
That means, for every $(\tilde{y}_1,\dots,\tilde{y}_d) \in S_\Psi^d$ we have
\[
H_d^\Psi(\tilde{y}_1,\dots,\tilde{y}_d) = \sqrt{1-\beta\gamma} \cdot F_d^\Phi(\tilde{x}_1,\dots,\tilde{x}_d)
\]
where $\tilde{x}_i = \tilde{y}_i / \sqrt{1-\beta\gamma}$ for $1\le i\le d$. 
Then, for each $i$, 
\[
\frac{\partial}{\partial \tilde{y}_i} H_d^\Psi(\tilde{y}_1,\dots,\tilde{y}_d) = \sqrt{1-\beta\gamma} \cdot \frac{\partial}{\partial \tilde{x}_i} F_d^\Phi(\tilde{x}_1,\dots,\tilde{x}_d) \cdot \frac{\mathrm{d} \tilde{x}_i}{\mathrm{d} \tilde{y}_i} 
= \frac{\partial}{\partial \tilde{x}_i} F_d^\Phi(\tilde{x}_1,\dots,\tilde{x}_d). 
\]
This implies that $\grad H_d^\Psi(\tilde{y}_1,\dots,\tilde{y}_d) = \grad F_d^\Phi(\tilde{x}_1,\dots,\tilde{x}_d)$ for all $(\tilde{y}_1,\dots,\tilde{y}_d) \in S_\Psi^d$, and the lemma then follows from \cref{thm:potentialgradientbound}. 
\end{proof}

\section{Remaining antiferromagnetic cases\texorpdfstring{: $\sqrt{\beta \gamma} \le \frac{\Delta-2}{\Delta}$ and $\gamma > 1$}{}}
\label{sec:antiferroweighted}

In this section, we discuss the case where $\sqrt{\beta \gamma} \le \frac{\Delta-2}{\Delta}$ and $\gamma > 1$. 
As studied in \cite{LLY13}, in this case the uniqueness region is more complicated. 
For example, there exists a critical $\lambda_c^* >0$ such that the $2$-spin system with $\lambda < \lambda_c^*$ is in the uniqueness region for arbitrary graphs; namely, $(\beta,\gamma,\lambda)$ is up-to-$\infty$ unique. 
To deal with large degrees, we need to relax the \hyperlink{cond:boundedness}{Boundedness} condition in \cref{defn:potential} and define a more general version of $(\alpha,c)$-potentials. 
We shall see that \cref{thm:contraction-implies-mixing} still holds for this general $(\alpha,c)$-potential. 
The reason behind it is that in order to bound the maximum eigenvalue of the influence matrix, it suffices to consider a vertex-weighted sum of absolute influences of a vertex with large degree. 

\begin{rmk}
We give more background on the uniqueness region in \cref{sec:uniqueness-background}. 
Note that in a recent revision of \cite{LLY13}, the authors updated the descriptions of the uniqueness region for the case $\sqrt{\beta \gamma} \le \frac{\Delta-2}{\Delta}$ and $\gamma > 1$, fixing a small error in the previous version. 
Statements and proofs in this section and \cref{sec:boundedness} of this paper are also adjusted accordingly based on the new version of \cite{LLY13}. 
\end{rmk}


Recall that our goal is to bound the maximum eigenvalue of the matrix $\II_G^\sL$. 
We can do this by upper bounding the absolute row sum $\sum_{v\in V \backslash \Lambda} |\II_G^\sL(r \sra v)|$ for fixed $r$, thereby giving us a valid upper bound on $\lambda_{\max}(\II_G^\sL)$. 
However, this approach does not work when $\sqrt{\beta \gamma} \le \frac{\Delta-2}{\Delta}$ and $\gamma > 1$. In this case, the potential $\Psi$ fails to be an $(\alpha,c)$-potential for a universal constant $c$ independent of $\Delta$. 
In fact, no such $(\alpha,c)$-potentials exist 
as the absolute row sum $\sum_{v\in V \backslash \Lambda} |\II_G^\sL(r \sra v)|$ can be as large as $\Theta(\Delta)$. 
Especially, if the parameters $(\beta,\gamma,\lambda)$ are up-to-$\infty$ unique, which means the spin system has uniqueness for arbitrary graphs, then the absolute row sum $\sum_{v\in V \backslash \Lambda} |\II_G^\sL(r \sra v)|$ can be $\Theta(n)$ where $n=|V|$. 
We give a specific example where this is the case.


\begin{ex}
Consider the antiferromagnetic 2-spin system specified by parameters $\beta=0$, $\gamma > 1$ and $\lambda>0$ on the star graph centered at $r$ with $\Delta$ leaves. 
A simple calculation reveals that $\abs{\mathcal{I}_{G}(r \sra v)} = \frac{\lambda}{\lambda+\gamma}$ 
for any leaf vertex $v \neq r$. Hence, $\sum_{v \neq r} \abs{\mathcal{I}_{G}(r \sra v)} = \Delta \cdot \frac{\lambda}{\lambda + \gamma}$. Now, since $\gamma > 1$, we have 
\[
\lambda_c = 
\lambda_{c}(\gamma,\Delta) = \min_{1 < d < \Delta} \frac{\gamma^{d+1}d^{d}}{(d-1)^{d+1}} = \Theta_{\gamma}(1), 
\]
forcing $\sum_{v \neq r} \abs{\mathcal{I}_{G}(r \sra v)} = \Theta_{\gamma}(\Delta)$ even when $\lambda < \lambda_c$ lies in the uniqueness region. 
However, we still have $\lambda_{\max}(\II_G) = O(1)$ since $\sum_{v \neq r} |\II_G(v \sra r)| = O(1)$. 
\end{ex}

To solve this issue, one might want to consider the absolute column sum, involving the sum of absolute influences on a fixed vertex. 
However, this will not allow us to use the beautiful connection between graphs and SAW trees as showed in \cref{lem:Tsaw-G}. 
Instead, we consider here a vertex-weighted version of the absolute row sum of $\II_G^\sL$, which also upper bounds the maximum eigenvalue.

\begin{lem}\label{lem:w-row-sum}
Let $\rho: V \to \mathbb{R}^+$ be a positive weight function of vertices. 
If there is a constant $\xi > 0$ such that for every $r\in V$ we have
\begin{equation}
\sum_{v\in V\backslash \Lambda} \rho_v \cdot \left|\II_G^\sL(r \sra v) \right| \le \xi \cdot \rho_r,
\end{equation}
then $\lambda_\mathrm{max}(\II_G^\sL) \le \xi$. 
\end{lem}
\begin{proof}
Let $\mathcal{P} = \mathrm{diag}\{\rho_v: v\in V\backslash \Lambda\}$. Then the assumption is equivalent to $\|\mathcal{P}^{-1} \II_G^\sL \mathcal{P}\|_\infty \le \xi$. 
It follows that $\lambda_{\max}(\II_G^\sL) = \lambda_{\max}(\mathcal{P}^{-1}\II_G^\sL \mathcal{P}) \le \xi$. 
\end{proof}

We then modify our definition of $(\alpha,c)$-potentials from \cref{defn:potential} which allows a weaker \hyperlink{cond:boundedness}{Boundedness} condition. 
We remark that the only two differences between \cref{defn:potential-weighted} and \cref{defn:potential} is that: we allow $\Delta = \infty$; and the \hyperlink{cond:boundedness}{Boundedness} condition is relaxed to what we call \hyperlink{cond:boundedness'}{General Boundedness}. 
Recall that for every $0\le d < \Delta$, we let $J_{d} = \wrapb{\log(\lambda \beta^{d}), \log(\lambda / \gamma^{d})}$ when $\beta\gamma < 1$, and $J_{d} = \wrapb{\log(\lambda / \gamma^{d}), \log(\lambda \beta^{d})}$ when $\beta\gamma > 1$. 

\begin{defn}[General $(\alpha,c)$-potential function]
\label{defn:potential-weighted}
Let $\Delta \ge 3$ be an integer or $\Delta = \infty$. 
Let $\beta,\gamma,\lambda$ be reals such that $0\le \beta \le \gamma$, $\gamma>0$ and $\lambda>0$. 
Let $\Psi:[-\infty,+\infty] \to (-\infty,+\infty)$ be a differentiable and increasing function with image $S = \Psi[-\infty,+\infty]$ and derivative $\psi = \Psi'$. 
For any $\alpha\in(0,1)$ and $c>0$, 
we say $\Psi$ is a \emph{general $(\alpha,c)$-potential function} with respect to $\Delta$ and $(\beta, \gamma, \lambda)$ if it satisfies the following conditions: 
\begin{enumerate}
\item \linkdest{cond:contraction}(Contraction) For every integer $d$ such that $1\le d<\Delta$ and every $(\tilde{y}_1,\dots,\tilde{y}_d) \in S^d$, we have
\[
\norm{\grad H_{d}^{\Psi}(\tilde{y}_1,\dots,\tilde{y}_d)}_1 = \sum_{i=1}^d \frac{\psi(y)}{\psi(y_i)} \cdot |h(y_i)| \le 1-\alpha
\]
where $H_{d}^{\Psi} = \Psi \circ H_{d} \circ \Psi^{-1}$, $y_i = \Psi^{-1}(\tilde{y}_i)$ for $1\le i \le d$, and $y = H_{d}(y_1,\dots,y_d)$. 
\item \linkdest{cond:boundedness'}(General Boundedness) 
For all integers $d_{1},d_{2}$ such that $0\leq d_1,d_2 < \Delta$, 
and all reals $y_{1} \in J_{d_{1}}, y_{2} \in J_{d_{2}}$, we have
\[
\frac{\psi(y_2)}{\psi(y_1)} \cdot |h(y_1)| \le \frac{2c}{d_1+d_2+2}. 
\]
\end{enumerate}
\end{defn}
Notice that \hyperlink{cond:boundedness'}{General Boundedness} is a weaker condition than \hyperlink{cond:boundedness}{Boundedness}. To see this, if a potential function $\Psi$ satisfies \hyperlink{cond:boundedness}{Boundedness} with parameter $c$, 
then for every $0\le d_i <\Delta$ and every $y_i \in J_{d_i}$ where $i=1,2$ we have
\[
\frac{\psi(y_2)}{\psi(y_1)} \cdot |h(y_1)| \le \frac{c}{\Delta} \le \frac{2c}{d_1+d_2+2}. 
\]
The following theorem generalizes \cref{thm:contraction-implies-mixing} and shows that 
a general $(\alpha,c)$-potential function is sufficient to establish rapid mixing of the Glauber dynamics. 


\begin{thm}
\label{thm:contraction-mixing-w}
Let $\Delta \ge 3$ be an integer or $\Delta = +\infty$. 
Let $\beta,\gamma,\lambda$ be reals such that $0\le \beta \le \gamma$, $\gamma>0$ and $\lambda>0$. 
Suppose that there is a general $(\alpha,c)$-potential with respect to $\Delta$ and $(\beta,\gamma,\lambda)$ for some $\alpha \in (0,1)$ and $c>0$. 
Then for every $n$-vertex graph $G$ of maximum degree $\Delta$, 
the mixing time of the Glauber dynamics for the $2$-spin system on $G$ with parameters $(\beta,\gamma,\lambda)$ is $O(n^{2+2c/\alpha})$. 
\end{thm}

We then give a counterpart of \cref{lem:potential-is-good}, 
showing that $\Psi$ is a general $(\alpha,c)$-potential when $\sqrt{\beta \gamma} \le \frac{\Delta-2}{\Delta}$ and $\gamma > 1$. 
\cref{thm:main-2-spin} for this case is then obtained from \cref{thm:contraction-mixing-w} and \cref{lem:potential-is-good-lambdac}.

\begin{lem}
\label{lem:potential-is-good-lambdac}
Let $\Delta \ge 3$ be an integer. 
Let $\beta,\gamma,\lambda$ be reals such that $0\le \beta < 1 < \gamma$ and $\sqrt{\beta \gamma} \le \frac{\Delta-2}{\Delta}$. 
Assume that $(\beta,\gamma,\lambda)$ is up-to-$\Delta$ unique with gap $\delta \in (0,1)$. 
Then the function $\Psi$ defined implicitly by \cref{eq:Psi} is a general $(\alpha,c)$-potential function with $\alpha \ge \delta/2$ and $c \leq 18$; 
we can further take $c \leq 4$ if $\beta = 0$. 
\end{lem}

The proof of \cref{thm:contraction-mixing-w} can be found in \cref{sec:weighted-contraction-mixing}. 
For \cref{lem:potential-is-good-lambdac}, the \hyperlink{cond:contraction}{Contraction} condition of $\Psi$ follows from \cref{lem:contraction-Psi}, and \hyperlink{cond:boundedness'}{General Boundedness} is proved in \cref{sec:boundedness} together with all other cases.

\section{Ferromagnetic cases}
\label{sec:ferro}
In the ferromagnetic case, the best known correlation decay results are given in \cite{GL18, SS19}. Using the potential functions in \cite{GL18} and \cite{SS19}, we show the following two results, which match the known correlation decay results. In fact, the potential function from \cite{SS19} turns out to be an $(\alpha,c)$-potential function for constants $\alpha = \Theta(\delta)$ and $c \leq O(1)$.
\begin{thm}\label{thm:ss19ferroregion}
Fix an integer $\Delta \geq 3$, positive real numbers $\beta,\gamma,\lambda$ and $0 < \delta < 1$, and assume $(\beta,\gamma,\lambda)$ satisfies one of the following three conditions:
\begin{enumerate}
    \item $\frac{\Delta-2+\delta}{\Delta-\delta} \leq \sqrt{\beta\gamma} \leq \frac{\Delta-\delta}{\Delta-2+\delta}$, and $\lambda$ is arbitrary;
    \item $\sqrt{\beta\gamma} \geq \frac{\Delta}{\Delta-2}$ and $\lambda \leq (1 - \delta) \frac{\gamma}{\max\{1,\beta^{\Delta-1}\} \cdot ((\Delta-2)\beta\gamma - \Delta)}$;
    \item $\sqrt{\beta\gamma} \geq \frac{\Delta}{\Delta-2}$ and $\lambda \geq \frac{1}{1 - \delta} \cdot \frac{(\Delta-2)\beta\gamma-\Delta}{\beta \cdot \min\{1,1/\gamma^{\Delta-1}\}}$.
\end{enumerate}
Then the identity function $\Psi(y) = y$ (based on the potential given in \cite{SS19}) is an $(\alpha,c)$-potential function for $\alpha = \Theta(\delta)$ and $c \leq O(1)$. Furthermore, for every $n$-vertex graph $G$ of maximum degree at most $\Delta$, the mixing time of the Glauber dynamics for the 2-spin system on $G$ with parameters $(\beta,\gamma,\lambda)$ is $O(n^{2+c/\delta})$, for a universal constant $c > 0$.
\end{thm}
\begin{remark}\label{rmk:better-contraction}
Condition 1 includes both the ferromagnetic case $1< \sqrt{\beta\gamma} \le \frac{\Delta-\delta}{\Delta-2+\delta}$ and the antiferromagnetic case $\frac{\Delta-2+\delta}{\Delta-\delta} \le \sqrt{\beta\gamma} < 1$. 
Note that in both cases $(\beta,\gamma,\lambda)$ is up-to-$\Delta$ unique with gap $\delta$. 
For the antiferromagnetic case, the identity function $\Psi$ is an $(\alpha,c)$-potential with $c\le 1.5$ and a better contraction rate $\alpha \ge \delta$, compared with the bound $\alpha \ge \delta/2$ of the potential $\Psi$ given by \cref{eq:Psi} in \cref{lem:potential-is-good}. 
For the ferromagnetic case with $\beta = \gamma>1$ (Ising model), a stronger result by \cite{MS13} was known, which gives $O(n\log n)$ mixing. 
\end{remark}

The potential function from \cite{GL18} is indeed an $(\alpha,c)$-potential, but $c$ must, unfortunately, depend on $\Delta$. We have the following result, which is weaker than the correlation decay algorithm in \cite{GL18} for unbounded degree graphs.
\begin{thm}\label{thm:gl18ferroregion}
Fix an integer $\Delta \geq 3$, and nonnegative real numbers $\beta,\gamma,\lambda$ satisfying $\beta \leq 1 \leq \gamma$, $\sqrt{\beta\gamma} \geq \frac{\Delta}{\Delta-2}$, and $\lambda < \wrapp{\frac{\gamma}{\beta}}^{\frac{\sqrt{\beta\gamma}}{\sqrt{\beta\gamma} - 1}}$. Then for every $n$-vertex graph $G$ with maximum degree at most $\Delta$, the mixing time of the Glauber dynamics for the ferromagnetic 2-spin system on $G$ with parameters $(\beta,\gamma,\lambda)$ is $O(n^{C})$, for a constant $C$ depending only on $\beta,\gamma,\lambda,\Delta$, but not $n$.
\end{thm}
Proofs of these theorems are provided in \cref{sec:ferroproofs}.

\printbibliography
\pagebreak

\begin{appendices}
\crefalias{section}{appsec}
\crefalias{subsection}{appsec}

\section{Proof of main results}
\label{sec:proof-main}

In this section we give the proofs of \cref{thm:main-hardcore}, \cref{thm:main-Ising}, \cref{thm:main-2-spin} and \cref{thm:contraction-implies-mixing}. 

\begin{proof}[Proof of \cref{thm:contraction-implies-mixing}]
Note that since the transition matrix $P$ for the Glauber dynamics has all nonnegative eigenvalues, we have that $\lambda^*(P) = 1 - \lambda_{2}(P)$ and so in order to deduce mixing, it suffices to lower bound $1 - \lambda_{2}(P)$. 
We do this by employing \cref{thm:spectral-influence}. It suffices to show $(\eta_{0},\dots,\eta_{n-2})$-spectrally independence for sufficiently small $\eta_{i}$.

To bound $\eta_{i}$, it suffices to bound $\sum_{v \in V \backslash \{r\}} \abs{\mathcal{I}_{G}^{\sigma_{\Lambda}}(r \sra v)}$ for all graphs $G = (V,E)$ with $n = |V|$ vertices and all boundary conditions $\sigma_{\Lambda}$ on a subset $\Lambda$ of $i$ vertices. We claim the following:
\begin{equation}
\label{eq:influence-main}
    \sum_{v \in V \backslash \{r\}} \abs{\mathcal{I}_{G}^{\sigma_{\Lambda}}(r \sra v)} \leq \min\wrapc{\frac{c}{\alpha}, C(n-i-1)}
\end{equation}
where $C \in (0,1)$ is a constant depending only on $\beta,\gamma,\lambda,\Delta$. 
The first upper bound $\frac{c}{\delta}$ is deduced by 
\begin{align*}
    \sum_{v \in V \backslash \{r\}} \abs{\mathcal{I}_{G}^{\sigma_{\Lambda}}(r \sra v)} 
    &\leq \sum_{v \in V_T \backslash \{r\}} \abs{\mathcal{I}_{T}^{\sigma_{\Lambda}}(r \sra v)} \tag*{(\cref{lem:Tsaw-G}; $T = \TSAW(G,r)$)} \\
    &= \sum_{k=1}^{\infty} \sum_{v \in L_{r}(k)} \abs{\mathcal{I}_{T}^{\sigma_{\Lambda_{\Lambda}}}(r \sra v)} \tag*{(split the sum by levels)} \\
    &\leq c \sum_{k=1}^{\infty} (1 - \alpha)^{k-1} \tag*{(\cref{lem:bound-for-tree})} \\
    &= \frac{c}{\alpha}.
\end{align*}
The second upper bound $C(n-i-1)$ is more trivial. 
Intuitively, it means each absolute pairwise influence $\abs{\mathcal{I}_{G}^{\sigma_{\Lambda}}(r \sra v)}$ is at most some constant $C$ and hence the sum of absolute influences is upper bounded by $C(n-i-1)$. 
The following two claims, whose proofs are provided in \cref{subsec:constantdimensions}, give a more precise statement. 
\begin{clm}[Antiferromagnetic Case]\label{lem:antiferroconstantsize}
Fix an integer $\Delta \geq 3$ and real numbers $\beta,\gamma,\lambda$, and assume $0\le \beta \leq \gamma$, $\gamma > 0$, $\beta \gamma < 1$ and $\lambda > 0$. 
Then for every $n$-vertex graph $G$ of maximum degree at most $\Delta$, the antiferromagnetic $2$-spin system on $G$ with parameters $(\beta,\gamma,\lambda)$ is $Cn$-spectrally independent, for a constant $0 < C < 1$ depending only on $\beta,\gamma,\lambda,\Delta$. Furthermore, if $(\beta,\gamma,\Delta)$ is up-to-$\Delta$ unique, then we can drop the dependence on $\Delta$.
\end{clm}
\begin{clm}[Ferromagnetic Case]\label{lem:ferroconstantsize}
Fix an integer $\Delta \geq 3$ and positive real numbers $\beta,\gamma,\lambda$, and assume $\beta \leq \gamma$ and $\beta \gamma > 1$. 
Then for every $n$-vertex graph $G$ of maximum degree at most $\Delta$, the ferromagnetic $2$-spin system on $G$ with parameters $(\beta,\gamma,\lambda)$ is $Cn$-spectrally independent, for a constant $0 < C < 1$ depending only on $\beta,\gamma,\lambda,\Delta$.
\end{clm}

With \cref{eq:influence-main} in hand, we immediately see that by \cref{thm:spectral-influence}, 
\[
1 - \lambda_{2}(P) 
\geq \frac{1}{n} \,\prod_{i=0}^{n-2} \wrapp{1 - \frac{\eta_{i}}{n-i-1}} 
\geq \frac{1}{n} \cdot (1 - C)^{2\lceil c/\alpha \rceil -1} \cdot 
\prod_{i=0}^{n-2\lceil c/\alpha \rceil-1} \wrapp{1 - \frac{c}{\alpha} \cdot \frac{1}{n-i-1}}.
\]
Using the fact that $1-x \ge \exp(-x-x^2)$ for all $0\le x \le \frac{1}{2}$ (which can be proved straightforwardly by calculus), we get
\begin{align*}
\prod_{i=0}^{n-2\lceil c/\alpha \rceil-1} \wrapp{1 - \frac{c}{\alpha} \cdot \frac{1}{n-i-1}} 
= 
\prod_{j= 2\lceil c/\alpha \rceil}^{n-1} \wrapp{1 - \frac{c}{\alpha} \cdot \frac{1}{j}} 
\ge \exp \left( - \frac{c}{\alpha} \sum_{j= 2\lceil c/\alpha \rceil}^{n-1} \frac{1}{j} - \frac{c^2}{\alpha^2} \sum_{j= 2\lceil c/\alpha \rceil}^{n-1} \frac{1}{j^2} \right). 
\end{align*}
Now since
\[
\sum_{j= 2\lceil c/\alpha \rceil}^{n-1} \frac{1}{j} \le \sum_{j=2}^n \frac{1}{j} \le \int_1^n \frac{d x}{x} = \log n
\]
and 
\[
\sum_{j= 2\lceil c/\alpha \rceil}^{n-1} \frac{1}{j^2} \le \sum_{j=2}^\infty \frac{1}{j(j-1)} = 1,
\]
we deduce that
\[
1 - \lambda_{2}(P) \ge (1 - C)^{2\lceil c/\alpha \rceil -1} \cdot e^{-(c/\alpha)^2} \cdot n^{-(1+c/\alpha)}. 
\]
The theorem then follows from \cref{eq:rel-mixing}. 
\end{proof}

\begin{proof}[Proof of \cref{thm:main-2-spin}]
We leverage \cref{thm:contraction-implies-mixing} and \cref{thm:contraction-mixing-w}, which shows $O(n^{2 + \frac{c}{\alpha}})$ mixing 
as long as there is an $(\alpha,c)$-potential, 
or $O(n^{2 + \frac{2c}{\alpha}})$ mixing if there is a general $(\alpha,c)$-potential. 
We use the potential given by \cref{eq:Psi}, which is an adaptation of the potential function in \cite{LLY13} to the log marginal ratios. 
When $(\beta,\gamma,\lambda)$ is up-to-$\Delta$ unique with gap $\delta \in (0,1)$, 
it is an $(\alpha,c)$-potential or a general $(\alpha,c)$-potential by \cref{lem:potential-is-good} and \cref{lem:potential-is-good-lambdac}, with $\alpha \geq \delta/2$ and $c$ a universal constant specified by the range of parameters. 
The theorem then follows. 
\end{proof}

\begin{proof}[Proof of \cref{thm:main-hardcore}]
By \cref{lem:hardcoreparamgap} later in \cref{sec:gapped}, $\lambda \leq (1 - \delta)\lambda_{c}(\Delta)$ implies up-to-$\Delta$ uniqueness with gap $\geq\delta/4$. Since $\gamma \leq 1$, we can again appeal to \cref{lem:potential-is-good} to obtain an $(\alpha,c)$-potential with $\alpha\geq \delta/8$ and $c \leq 4$. \cref{thm:main-hardcore} then follows by \cref{thm:contraction-implies-mixing} with $O(n^{2 +32/\delta})$ mixing.
\end{proof}

\begin{proof}[Proof of \cref{thm:main-Ising}]
By \cref{lem:largebetagammaparamgap} later in \cref{sec:gapped}, $\beta \geq \beta_{c}(\Delta) + \delta(1 - \beta_{c}(\Delta))$ implies up-to-$\Delta$ uniqueness with gap $\delta$. Again, appealing to \cref{lem:potential-is-good}, we obtain an $(\alpha,c)$-potential with $\alpha \geq \delta/2$ and $c \leq 1.5$. \cref{thm:main-Ising} then follows by \cref{thm:contraction-implies-mixing} with $O(n^{2 + 3/\delta})$ mixing. 

Though we technically get $O(n^{2+3/\delta})$ by using the \cite{LLY13} potential, we can improve it to $O(n^{2+1.5/\delta})$ mixing by using the trivial identity function as the potential. 
See the first case of \cref{thm:ss19ferroregion} (proved in \cref{sec:ferro_proof_SS19}) and \cref{rmk:better-contraction}. 
\end{proof}

\subsection{Uniqueness gaps in terms of parameter paps}
\label{sec:gapped}

In this section we state and prove \cref{lem:hardcoreparamgap} and \cref{lem:largebetagammaparamgap}, which relate the parameter gaps with the uniqueness gaps. 

\begin{clm}[Hardcore Model; Lemma C.1 from \cite{ALO20}]\label{lem:hardcoreparamgap}
Fix an integer $\Delta \geq 3$, $0 < \delta < 1$, and $\beta = 0, \gamma > 0$. If $\lambda \leq (1 - \delta)\lambda_{c}(\gamma,\Delta)$, then $(\beta,\gamma,\lambda)$ is up-to-$\Delta$ unique with gap $\delta/4$.
\end{clm}

\begin{clm}[Large $\sqrt{\beta\gamma}$]\label{lem:largebetagammaparamgap}
Fix an integer $\Delta \geq 3$, and $0 < \delta < 1$. If $\sqrt{\beta\gamma} \geq \frac{\Delta-2}{\Delta} + \delta\wrapp{1 - \frac{\Delta-2}{\Delta}} = \frac{\Delta - 2(1-\delta)}{\Delta}$, then $(\beta,\gamma,\lambda)$ is up-to-$\Delta$ unique with gap $0 < \delta < 1$ for all $\lambda$. Note if $\beta = \gamma$, this is precisely the condition $\beta \geq \beta_{c}(\Delta) + \delta(1 - \beta_{c}(\Delta))$.
\end{clm}
\begin{proof}
Consider the univariate recursion for the marginal ratios with $d < \Delta$ children $f_{d}(R) = \lambda \wrapp{\frac{\beta R + 1}{R + \gamma}}^{d}$. Differentiating, we have
\begin{align*}
    f_{d}'(R) &= d\lambda \wrapp{\frac{\beta R + 1}{R + \gamma}}^{d-1} \cdot \wrapp{\frac{\beta}{R + \gamma} - \frac{\beta R + 1}{(R + \gamma)^{2}}} = -d(1 - \beta\gamma) \lambda \wrapp{\frac{\beta R + 1}{R + \gamma}}^{d} \cdot \frac{1}{(\beta R + 1)(R + \gamma)} \\
    &= -d (1 - \beta\gamma) \cdot \frac{f_{d}(R)}{(\beta R + 1)(R + \gamma)}.
\end{align*}
At the unique fixed point $R_{d}^{*}$, we have $f_{d}(R_{d}^{*}) = R_{d}^{*}$ so
\begin{align*}
    \abs{f_{d}'(R_{d}^{*})} = d(1 - \beta\gamma) \frac{R_{d}^{*}}{(\beta R_{d}^{*} + 1)(R_{d}^{*} + \gamma)}.
\end{align*}
By \cref{lem:psimaximizer}, 
we have the upper bound
\begin{align*}
    \abs{f_{d}'(R_{d}^{*})} \leq d \cdot \frac{1 - \beta\gamma}{(1 + \sqrt{\beta\gamma})^{2}} = d \cdot \frac{1 - \sqrt{\beta\gamma}}{1 + \sqrt{\beta\gamma}}.
\end{align*}
Since we assumed $\sqrt{\beta\gamma} \geq \frac{\Delta-2(1-\delta)}{\Delta}$, we obtain
\begin{align*}
    d \cdot \frac{1 - \sqrt{\beta\gamma}}{1 + \sqrt{\beta\gamma}} \leq d \cdot \frac{\Delta - (\Delta - 2(1-\delta))}{\Delta + (\Delta - 2(1-\delta))} = d \cdot \frac{1-\delta}{\Delta - 1 + \delta} \leq (1-\delta)\frac{d}{\Delta-1}.
\end{align*}
As this is at most $1 - \delta$ for all $d < \Delta$, we have up-to-$\Delta$ uniqueness with gap $\delta$.
\end{proof}

    

\subsection{Spectral independence bounds for constant-size graphs}\label{subsec:constantdimensions}

In this section, we prove spectral independence bounds for graphs with fewer than $O(c/\alpha)$-many vertices, since for graphs with such few vertices, our bounds based on contraction of the tree recursions become trivial.
\begin{proof}[Proof of \cref{lem:antiferroconstantsize}]
If $R_{v}$ denotes the marginal ratio of a vertex $v \in G$, then $R_{v} \geq \lambda \beta^{\Delta}$. In the case $\gamma \leq 1$, we have $R_{v} \leq \lambda/\gamma^{\Delta}$ as well; if $\gamma > 1$, we have $R_{v} \leq \lambda$. It follows that we immediately have the bounds
\begin{align*}
    \abs{\mathcal{I}_{G} (u \sra v)} \leq \begin{cases}
        \abs{\frac{\lambda}{\lambda + \gamma^{\Delta}} - \frac{\lambda\beta^{\Delta}}{1 + \lambda \beta^{\Delta}}} = \frac{\lambda(1 - \beta^{\Delta}\gamma^{\Delta})}{(\lambda + \gamma^{\Delta})(1 + \lambda \beta^{\Delta})}, &\quad\text{if } \gamma \leq 1 \\
        \abs{\frac{\lambda}{1 + \lambda} - \frac{\lambda \beta^{\Delta}}{1 + \lambda\beta^{\Delta}}} = \frac{\lambda (1 - \beta^{\Delta})}{(\lambda + 1)(1 + \lambda \beta^{\Delta})}, &\quad\text{o.w.}
    \end{cases}
\end{align*}
for all $u,v \in G$. Note that these constants are less than $1$, and only depend on $\beta,\gamma,\lambda,\Delta$, yielding the first claim.

Now, we proceed to remove the dependence on $\Delta$ when up-to-$\Delta$ uniqueness holds. We have the following cases:
\begin{enumerate}
    \item If $\gamma > 1$, we immediately obtain a bound of $\frac{\lambda}{1 + \lambda}$ which is independent of $\Delta$.
    \item If $\beta = 0$ and $\gamma \leq 1$, then $\frac{\lambda(1 - \beta^{\Delta}\gamma^{\Delta})}{(\lambda + \gamma^{\Delta})(1 + \lambda \beta^{\Delta})} = \frac{\lambda}{\lambda + \gamma^{\Delta}} \leq \frac{\lambda}{\gamma^{\Delta}}$. Since $(\beta,\gamma,\lambda)$ is up-to-$\Delta$ unique, we must have $\lambda \leq \lambda_{c}(\gamma,\Delta) = \min_{1 < d < \Delta} \frac{\gamma^{d+1}d^{d}}{(d-1)^{d+1}} \leq \frac{\gamma^{\Delta}(\Delta-1)^{\Delta-1}}{(\Delta-2)^{\Delta}} \leq \gamma^{\Delta} \cdot O(1/\Delta)$. It follows that $\frac{\lambda}{\gamma^{\Delta}} \leq O(1/\Delta)$.
    \item If $\sqrt{\beta\gamma} > \frac{\Delta-2}{\Delta}$ and $\gamma \leq 1$, then
    \begin{align*}
        \frac{\lambda (1 - \beta^{\Delta}\gamma^{\Delta})}{(\lambda + \gamma^{\Delta})(1 + \lambda \beta^{\Delta})} \leq 1 - \beta^{\Delta}\gamma^{\Delta} \approx 1 - e^{-2}.
    \end{align*}
    \item If $\sqrt{\beta\gamma} \leq \frac{\Delta-2}{\Delta}$, then let $\Delta_{0}$ be the maximal $1 < d < \Delta$ such that $\sqrt{\beta\gamma} > \frac{d-2}{d}$. If $\lambda \leq \lambda_{c}(\beta,\gamma,\Delta)$, then by \cref{lem:twocriticalthresholdbounds}, we have
    \begin{align*}
        \frac{\lambda (1 - \beta^{\Delta}\gamma^{\Delta})}{(\lambda + \gamma^{\Delta})(1 + \lambda \beta^{\Delta})} \leq \frac{\lambda}{\gamma^{\Delta}} \leq O(\Delta_{0}/\Delta).
    \end{align*}
    If $\lambda \geq \overline{\lambda}_{c}(\beta,\gamma,\Delta)$, then again by \cref{lem:twocriticalthresholdbounds}, we have
    \[
        \frac{\lambda (1 - \beta^{\Delta}\gamma^{\Delta})}{(\lambda + \gamma^{\Delta})(1 + \lambda \beta^{\Delta})} \leq \frac{1}{\lambda \beta^{\Delta}} \leq O(\Delta_{0}/\Delta).\qedhere
    \]
\end{enumerate}
\end{proof}
\begin{proof}[Proof of \cref{lem:ferroconstantsize}]
The proof is identical to the antiferromagnetic case and we omit it here.
\end{proof}

\section{Proof of \texorpdfstring{\cref{lem:2spin-sys-property}}{Lemma 12} (Parts 1 and 2)}
\label{sec:proof-2spin}
\begin{proof}[Proof of \cref{lem:2spin-sys-property} (Parts 1 and 2)]
To see the first equality, we compute directly and get
\begin{align*}
\left(\optheta{v}\right) \log Z^\sL_G 
&= \frac{1}{Z^\sL_G} \cdot \left(\optheta{v}\right) Z^\sL_G\\ 
&= \frac{1}{Z^\sL_G} 
\sum_{\sigma \in \{0,1\}^{V\backslash \Lambda}} \left(\optheta{v}\right)
\left( \beta^{m_1(\sigma)} \gamma^{m_0(\sigma)} \prod_{w\in V} \lambda_w^{\sigma_w} \right)\\
&= \frac{1}{Z^\sL_G} 
\sum_{\sigma \in \{0,1\}^{V\backslash \Lambda}} \sigma_v
\left( \beta^{m_1(\sigma)} \gamma^{m_0(\sigma)} \prod_{w\in V} \lambda_w^{\sigma_w} \right)\\
&= \sum_{\sigma \in \{0,1\}^{V\backslash \Lambda}} \sigma_v \cdot \mu_G(\sigma \mid \sL)
= \marginal_G^\sL(v). 
\end{align*}
For Part 2, using the result above, we can also get
\begin{align*}
&\left(\optheta{v}\right) \left(\optheta{u}\right) \log Z^\sL_G\\ 
={}& \left(\optheta{v}\right) 
\left( \frac{1}{Z^\sL_G} \cdot \left(\optheta{u}\right) Z^\sL_G \right)\\
={}& \frac{1}{Z^\sL_G} \cdot \left(\optheta{v}\right) \left(\optheta{u}\right) Z^\sL_G 
- 
\frac{1}{(Z^\sL_G)^2} \cdot \left(\optheta{v}\right) Z^\sL_G \cdot \left(\optheta{u}\right) Z^\sL_G\\ 
={}& \frac{1}{Z^\sL_G} \cdot \left(\optheta{v}\right)
\left( \sum_{\sigma \in \{0,1\}^{V\backslash \Lambda}} \sigma_u
\left( \beta^{m_1(\sigma)} \gamma^{m_0(\sigma)} \prod_{w\in V} \lambda_w^{\sigma_w} \right) \right)
- \marginal_G^\sL(u) \cdot \marginal_G^\sL(v)\\
={}& \frac{1}{Z^\sL_G} \sum_{\sigma \in \{0,1\}^{V\backslash \Lambda}} \sigma_u \cdot
\left(\optheta{v}\right) 
\left( \beta^{m_1(\sigma)} \gamma^{m_0(\sigma)} \prod_{w\in V} \lambda_w^{\sigma_w} \right) 
- \marginal_G^\sL(u) \cdot \marginal_G^\sL(v)\\
={}& \frac{1}{Z^\sL_G} \sum_{\sigma \in \{0,1\}^{V\backslash \Lambda}} \sigma_u \cdot \sigma_v
\left( \beta^{m_1(\sigma)} \gamma^{m_0(\sigma)} \prod_{w\in V} \lambda_w^{\sigma_w} \right) 
- \marginal_G^\sL(u) \cdot \marginal_G^\sL(v)\\
={}& \sum_{\sigma \in \{0,1\}^{V\backslash \Lambda}} \sigma_u \cdot \sigma_v \cdot \mu_G(\sigma \mid \sL) 
- \marginal_G^\sL(u) \cdot \marginal_G^\sL(v)\\
={}& \CovM_G^\sL(u,v). \qedhere
\end{align*}
\end{proof}

\section{A technical lemma for \texorpdfstring{$\Psi$}{potential}}
\label{sec:well-def-potential}
The following lemma implies that the potential $\Psi$ given by \cref{eq:Psi} is well-defined. 

\begin{lem}
For all $\beta,\gamma>0$ such that $\beta\gamma<1$, we have 
\[
\int_{-\infty}^{+\infty} \sqrt{\frac{(1-\beta\gamma)e^{y}}{(\beta e^{y} + 1)(e^{y} + \gamma)}} < +\infty. 
\]
\end{lem}

\begin{proof}
For the $+\infty$ side we have 
\[
\int_{0}^{+\infty} \sqrt{\frac{(1-\beta\gamma)e^{y}}{(\beta e^{y} + 1)(e^{y} + \gamma)}} 
= \int_{0}^{+\infty} \sqrt{\frac{1-\beta\gamma}{\beta e^y + \gamma e^{-y} + \beta\gamma + 1}} 
< \int_{0}^{+\infty} \frac{1}{\sqrt{\beta e^y}} < +\infty.
\]
Similarly, for the $-\infty$ side we have 
\[
\int_{-\infty}^{0} \sqrt{\frac{(1-\beta\gamma)e^{y}}{(\beta e^{y} + 1)(e^{y} + \gamma)}} 
< \int_{-\infty}^{0} \frac{1}{\sqrt{\gamma e^{-y}}} < +\infty. \qedhere
\]
\end{proof}

\section{Mixing by the potential method: Proof of \texorpdfstring{\cref{thm:contraction-mixing-w}}{Theorem 24}}
\label{sec:weighted-contraction-mixing}

In this section, we prove \cref{thm:contraction-mixing-w} in the same way of \cref{thm:contraction-implies-mixing}, as outlined in \cref{sec:intro-proof}. 
The major difference here is that we consider a weighted sum of absolute influences $\sum_{v\in V\backslash \Lambda} \rho_v \cdot \left|\II_G^\sL(r \sra v) \right|$ where $\rho: V \to \R^+$ is a weight function. 
This is sufficient for us to bound the eigenvalue of the influence matrix, as indicated by \cref{lem:w-row-sum}. 
We will choose the weight of a vertex $v$ to be $\rho_v = \Delta_v$, the degree of $v$. 
The following lemma provides us an upper bound on the weighted sum of absolute influences to distance $k$, given a general $(\alpha,c)$-potential. 
In particular, it generalizes \cref{lem:bound-for-tree}.

\begin{lem}
\label{lem:bound-for-tree-weighted}
If there exists a general $(\alpha,c)$-potential function $\Psi$ with respect to $\Delta$ and $(\beta,\gamma,\lambda)$ where $\alpha \in (0,1)$ and $c>0$, then for every $\Lambda \subseteq V_T \backslash \{r\}$, $\sL \in \{0,1\}^\Lambda$ and all integers $k \geq 1$, 
\begin{equation*}
    \sum_{v \in L_{r}(k)} \Delta_v \cdot \abs{\mathcal{I}_{T}^\sL(r\sra v)} \leq 2c \cdot (1 - \alpha)^{k-1} \cdot \Delta_r
\end{equation*}
where $L_r(k)$ denote the set of all free vertices at distance $k$ away from $r$. 
\end{lem}

To prove \cref{lem:bound-for-tree-weighted}, we first state the following generalization of \cref{lem:A-and-B} for any weight function $\rho$. The proof of \cref{lem:A-and-B-weighted} is identical to \cref{lem:A-and-B} and we omit here. 
\begin{lem}
\label{lem:A-and-B-weighted}
Let $\Psi: [-\infty,+\infty] \to (-\infty,+\infty)$ be a differentiable and increasing (potential) function with image $S = \Psi[-\infty,+\infty]$ and derivative $\psi = \Psi'$. 
Denote the degree of the root $r$ by $\Delta_r$. 
Then for every integer $k \ge 1$, 
\begin{equation*}
    \sum_{v \in L_r(k)} \rho_v \cdot \left| \mathcal{I}_{T}^\sL(r\sra v) \right|
    \le \Delta_r A_\Psi B_\Psi^\rho 
    \left( \max_{1\le d < \Delta} \,\sup_{\tilde{\mybf{y}} \in S^d} \norm{\grad H_{d}^{\Psi} (\tilde{\mybf{y}})}_1 \right)^{k-1}
\end{equation*}
where 
\[
A_\Psi = \max_{u\in L_r(1)} \left\{ \frac{|h(\log R_u)|}{\psi(\log R_u)} \right\} 
\quad\text{and}\quad 
B_\Psi^\rho = \max_{v\in L_r(k)} \left\{ \rho_v \cdot \psi(\log R_v) \right\}. 
\]
\end{lem}

We then prove \cref{lem:bound-for-tree-weighted} and \cref{thm:contraction-mixing-w}. 

\begin{proof}[Proof of \cref{lem:bound-for-tree-weighted}]
Denote the degree of a vertex $v\in V_T \backslash\{r\}$ by $\Delta_v$, and the degree of $v$ in the subtree $T_v$ by $d_v = \Delta_v -1$. 
Pick the weights of vertices to be $\rho_v = \Delta_v$ for all $v\in V_T$. 
Since $\Psi$ is a general $(\alpha,c)$-potential, the \hyperlink{cond:contraction}{Contraction} condition implies that 
\[
\max_{1\le d < \Delta} \sup_{\tilde{\mybf{y}} \in S^d} \norm{\grad H_{d}^{\Psi} (\tilde{\mybf{y}})}_1 \le 1-\alpha. 
\]
Since $\log R_v \in J_{d_v}$ by the definition of $J_d$, 
the \hyperlink{cond:boundedness'}{General Boundedness} condition implies that for all $u\in L_r(1)$ and $v\in L_r(k)$, 
\[
\frac{\psi(\log R_v)}{\psi(\log R_u)} \cdot |h(\log R_u)| \le \frac{2c}{\Delta_u + \Delta_v}. 
\]
Therefore, we get
\[
\Delta_{r} A_\Psi B_\Psi^\rho = \Delta_r \cdot \max_{u\in L_r(1)} \left\{ \frac{|h(\log R_u)|}{\psi(\log R_u)} \right\} \cdot \max_{v\in L_r(k)} \left\{ \Delta_v \cdot \psi(\log R_v) \right\} \le 2c \cdot \Delta_r.  
\]
The lemma then follows immediately from \cref{lem:A-and-B-weighted}. 
\end{proof}

\begin{proof}[Proof of \cref{thm:contraction-mixing-w}]
The proof of \cref{thm:contraction-mixing-w} is almost identical to \cref{thm:contraction-implies-mixing}. We point out that the only difference here is that we consider the weighted sum of absolute influences of a given vertex. Since the SAW tree preserve degrees of vertices, we can still apply \cref{lem:Tsaw-G}. Then, combining \cref{thm:spectral-influence}, \cref{lem:w-row-sum}, \cref{lem:Tsaw-G} and \cref{lem:bound-for-tree-weighted}, we complete the proof of the theorem. 
\end{proof}

\section{Verifying a good potential: Boundedness}
\label{sec:boundedness}
In this subsection, we show the \hyperlink{cond:boundedness}{Boundedness} or \hyperlink{cond:boundedness'}{General Boundedness} condition for our potential function $\Psi$ defined by \cref{eq:Psi} in different ranges of parameters. 
Combining \cref{lem:contraction-Psi}, we complete the proofs of \cref{lem:potential-is-good} and \cref{lem:potential-is-good-lambdac}. 

In \cref{sec:uniqueness-background} we give background on the uniqueness region of the parameters $(\beta,\gamma,\lambda)$, based on the work of \cite{LLY13}. 
We then show \hyperlink{cond:boundedness}{Boundedness} and \hyperlink{cond:boundedness'}{General Boundedness} in \cref{sec:boundedness-proof}. Proofs of technical lemmas are left to \cref{sec:boundedness-lemma}. 


\subsection{Preliminaries of the uniqueness region}
\label{sec:uniqueness-background}
In this section we give a brief description of the uniqueness region of parameters $(\beta,\gamma,\lambda)$. 
All the results here, and also their proofs, can be found in Lemma 21 from the latest version of \cite{LLY13}. 

Let $\Delta \ge 3$ be an integer and $\beta,\gamma,\lambda$ be reals. 
We assume that $0 \le \beta \le \gamma$, $\gamma > 0$, $\beta \gamma < 1$ and $\lambda > 0$. 
For $1 \leq d \le \Delta$ define
\[
f_{d}(R) = \lambda \wrapp{\frac{\beta R + 1}{R + \gamma}}^{d}
\]
and denote the unique fixed point of $f_{d}$ by $R_d^*$. 
Recall that the parameters $(\beta,\gamma,\lambda)$ are up-to-$\Delta$ unique with gap $\delta \in (0,1)$ if $|f'_d(R^*_d)| < 1- \delta$ for all $1 \leq d < \Delta$. 

When $\beta=0$, the spin system is called a \emph{hard-constraint model}. 
In this case, there exists a critical threshold for the external field defined as
\[
\lambda_c = \lambda_c(\gamma,\Delta) = \min_{1<d<\Delta} \frac{\gamma^{d+1} d^d}{(d-1)^{d+1}}, 
\]
such that the parameters $(0,\gamma,\lambda)$ are up-to-$\Delta$ unique if and only if $\lambda < \lambda_c$. 
In particular, when $\gamma \le 1$ the critical field is given by
\[
\lambda_c = \lambda_c(\gamma,\Delta) = \frac{\gamma^\Delta (\Delta-1)^{\Delta-1}}{(\Delta-2)^\Delta}. 
\]

When $\beta>0$, the spin system is called a \emph{soft-constraint model}. 
If $\sqrt{\beta\gamma} > \frac{\Delta-2}{\Delta}$, then $(\beta,\gamma,\lambda)$ is up-to-$\Delta$ unique for all $\lambda > 0$. 
If $\sqrt{\beta\gamma} \le \frac{\Delta-2}{\Delta}$
the uniqueness region is more complicated which we now describe.
Let
\[
\overline{\Delta} = \frac{1+\sqrt{\beta\gamma}}{1-\sqrt{\beta\gamma}} ,
\]
so that for every $1\le d < \overline{\Delta}$ we have $d \cdot \frac{1-\sqrt{\beta\gamma}}{1+\sqrt{\beta\gamma}} < 1$, and for every $d \ge \overline{\Delta}$ we have $d \cdot \frac{1-\sqrt{\beta\gamma}}{1+\sqrt{\beta\gamma}} \ge 1$. 
For every $\overline{\Delta} \le d < \Delta$, we define $x_1(d) \le x_2(d)$ to be the two positive roots of the quadratic equation
\[
\frac{d(1-\beta\gamma)x}{(\beta x+1)(x+\gamma)} = 1.
\]
More specifically, $x_1(d)$ and $x_2(d)$ are given by
\[
x_1(d) = \frac{\theta(d) - \sqrt{\theta(d)^2 - 4\beta\gamma}}{2\beta} 
\qquad\text{and}\qquad 
x_2(d) = \frac{\theta(d) + \sqrt{\theta(d)^2 - 4\beta\gamma}}{2\beta}
\]
where 
\[
\theta(d) = d(1-\beta\gamma) - (1+\beta\gamma). 
\]
Notice that $\theta(d) \ge 2\sqrt{\beta \gamma}$ for all $d\ge \overline{\Delta}$. 
For $i=1,2$ we let  
\[
\lambda_i(d) = x_i(d) \left( \frac{x_i(d) + \gamma}{\beta x_i(d) + 1} \right)^d. 
\]
Then, the parameters $(\beta,\gamma,\lambda)$ are up-to-$\Delta$ unique if and only if $\lambda$ belongs to the following regime
\begin{equation}\label{eq:regime-uniq}
\mathcal{A} = 
\bigcap_{\overline{\Delta} \le d < \Delta} \Big[ (0, \lambda_1(d)) \cup (\lambda_2(d), \infty) \Big].
\end{equation}
In particular, when $\gamma \le 1$ 
there are two critical thresholds $0 < \lambda_c < \overline{\lambda}_c$ such that the parameters $(\beta,\gamma,\lambda)$ are up-to-$\Delta$ unique if and only if $\lambda < \lambda_c$ or $\lambda > \overline{\lambda}_c$ (i.e., $\mathcal{A} = (0,\lambda_c) \cup (\overline{\lambda}_c,\infty)$), where
\[
\lambda_c = \lambda_c(\beta,\gamma,\Delta) = \min_{\overline{\Delta} \le d < \Delta} \lambda_1(d)
\qquad\text{and}\qquad 
\overline{\lambda}_c = \overline{\lambda}_c(\beta,\gamma,\Delta) = \max_{\overline{\Delta} \le d < \Delta} \lambda_2(d) = \lambda_2(\Delta-1). 
\]

The following bounds on the critical fields are helpful for our proofs later. 
\begin{lem}\label{lem:twocriticalthresholdbounds}
\begin{enumerate}
\item If $\beta = 0$, then for every integer $d$ such that $1 < d < \Delta$ we have
\[
\lambda_c \le \frac{4\gamma^{d+1}}{d-1}.
\]
\item If $\beta>0$ and $\sqrt{\beta \gamma} \le \frac{\Delta-2}{\Delta}$, then for every integer $d$ such that $\overline{\Delta} \le d < \Delta$ we have
\[
\lambda_1(d) \le \frac{18\gamma^{d+1}}{\theta(d)}
\qquad\text{and}\qquad
\lambda_2(d) \ge \frac{\theta(d)}{18\beta^{d+1}}
\]
where $\theta(d) = d(1-\beta\gamma) - (1+\beta\gamma)$. 
\end{enumerate}
\end{lem}

The proof of \cref{lem:twocriticalthresholdbounds} is postponed to \cref{sec:boundedness-lemma}. 


\subsection{Proofs of boundedness}
\label{sec:boundedness-proof}

In this section we complete the proofs of \cref{lem:potential-is-good} and \cref{lem:potential-is-good-lambdac} by establishing \hyperlink{cond:boundedness}{Boundedness} and \hyperlink{cond:boundedness'}{General Boundedness} in the corresponding range of parameters. 

Let $\Delta \ge 3$ be an integer. 
Let $\beta,\gamma,\lambda$ be reals such that $0\le \beta \le \gamma$, $\gamma>0$, $\beta\gamma <1$ and $\lambda>0$. 
Recall that the potential function $\Psi$ is defined by 
\begin{equation*}
    \Psi'(y) 
    = \psi(y) 
    = \sqrt{\frac{(1-\beta\gamma)e^{y}}{(\beta e^{y} + 1)(e^{y} + \gamma)}} 
    = \sqrt{\abs{h(y)}}, 
    \qquad
    \Psi(0) = 0. \tag{1}
\end{equation*}
It is surprising to find out that $\psi = \sqrt{|h|}$, as the potential $\Psi$ is exactly the one from \cite{LLY13} as indicated by \cref{lem:Psi-Phi}. 
This seems not to be a coincidence, and it provides some intuition why the potential from \cite{LLY13} works. 
More importantly, the fact that $\psi = \sqrt{|h|}$ is helpful in our proof of \hyperlink{cond:boundedness}{Boundedness} and \hyperlink{cond:boundedness'}{General Boundedness}. 
Recall that for $0\le d < \Delta$ and $\beta\gamma < 1$ we let $J_{d} = \wrapb{\log(\lambda \beta^{d}), \log(\lambda / \gamma^{d})}$ to be the range of log marginal ratios of a vertex with $d$ children. 
Then for every $0\le d_i < \Delta$ and $y_i \in J_{d_i}$ where $i=1,2$, we have
\begin{equation}\label{eq:psi1psi2}
\frac{\psi(y_2)}{\psi(y_1)} \cdot |h(y_1)| = \sqrt{|h(y_1)| \cdot |h(y_2)|}. 
\end{equation}
The following lemma gives upper bounds on $\sqrt{|h(y_1)| \cdot |h(y_2)|}$, from which and \cref{eq:psi1psi2} we deduce \hyperlink{cond:boundedness}{Boundedness} and \hyperlink{cond:boundedness'}{General Boundedness} immediately. 
The brackets in the lemma indicate which lemma the bound is applied to.

\begin{lem}\label{lem:bound-Psi}
Let $\Delta \ge 3$ be an integer. 
Let $\beta,\gamma,\lambda$ be reals such that $0\le \beta \le \gamma$, $\gamma>0$, $\beta\gamma <1$ and $\lambda>0$. 
Assume that the parameters $(\beta,\gamma,\lambda)$ are up-to-$\Delta$ unique with gap $\delta \in (0,1)$. 
Then for all integers $d_1,d_2$ such that $0\le d_1,d_2 < \Delta$, and all reals $y_i \in J_{d_i}$ where $i=1,2$, the following holds:
\begin{enumerate}
\item[H.] Hard-constraint models: $\beta = 0$ and $\lambda < \lambda_c$. 
\begin{enumerate}[{H.}1.]
\item (\cref{lem:potential-is-good}) If $\gamma \le 1$, then
\[
|h(y_1)| \le \frac{4}{\Delta}. 
\]
\item (\cref{lem:potential-is-good-lambdac}) If $\gamma > 1$, then
\[
\sqrt{|h(y_1)| \cdot |h(y_2)|} \le \frac{8}{d_1 + d_2 + 2}. 
\]
\end{enumerate} 
\item[S.] Soft-constraint models: $\beta > 0$ and $\lambda \in \mathcal{A}$.
\begin{enumerate}[{S.}1.]
\item (\cref{lem:potential-is-good}) If $\sqrt{\beta \gamma} > \frac{\Delta-2}{\Delta}$, then
\[
|h(y_1)| \le \frac{1.5}{\Delta}. 
\]
\item (\cref{lem:potential-is-good}) If $\sqrt{\beta \gamma} \le \frac{\Delta-2}{\Delta}$ and $\gamma \le 1$, then
\[
|h(y_1)| \le \frac{18}{\Delta}. 
\]
\item (\cref{lem:potential-is-good-lambdac}) If $\sqrt{\beta \gamma} \le \frac{\Delta-2}{\Delta}$ and $\gamma > 1$, then
\[
\sqrt{|h(y_1)| \cdot |h(y_2)|} \le \frac{36}{d_1 + d_2 + 2}. 
\]
\end{enumerate} 
\end{enumerate}
\end{lem}

The following lemma, whose proof can be found in \cref{sec:boundedness-lemma}, is helpful. 
\begin{lem}\label{lem:psimaximizer}
For every $y\in[-\infty,+\infty]$ we have
\[
|h(y)| = \frac{|1-\beta\gamma| e^y}{(\beta e^y+1)(e^y+\gamma)} \le \frac{|1-\sqrt{\beta\gamma}|}{1+\sqrt{\beta\gamma}}. 
\]
\end{lem}

We present here the proof of \cref{lem:bound-Psi}. 
\begin{proof}[Proof of \cref{lem:bound-Psi}]
We use notations and results from \cref{sec:uniqueness-background}. 

\medskip
\noindent
\emph{H.~~Hard-constraint models: $\beta = 0$ and $\lambda < \lambda_c$.}

\medskip
\noindent
\emph{H.1.~~$\gamma \le 1$.} 

\medskip
For every $y_1 \in J_{d_1}$ we deduce from \cref{lem:twocriticalthresholdbounds} that 
\[
e^{y_1} \le \frac{\lambda}{\gamma^{d_1}} \le \frac{\lambda_c}{\gamma^{\Delta-1}} \le \frac{4\gamma}{\Delta - 2}. 
\]
Hence, 
\[
|h(y_1)| = \frac{e^{y_1}}{e^{y_1}+\gamma} 
\le \frac{\frac{4\gamma}{\Delta - 2}}{\frac{4\gamma}{\Delta - 2}+\gamma} 
= \frac{4}{\Delta+2} 
\le \frac{4}{\Delta}. 
\]


\medskip
\noindent
\emph{H.2.~~$\gamma > 1$.} 

\medskip
Let $\bar{y} = \frac{y_1+y_2}{2}$ and $\bar{d} = \frac{d_1+d_2}{2}$. 
Then we get
\[
\sqrt{|h(y_1)| \cdot |h(y_2)|} = \sqrt{\frac{e^{y_1}}{e^{y_1}+\gamma}} \cdot \sqrt{\frac{e^{y_2}}{e^{y_2}+\gamma}} 
= \frac{1}{\sqrt{(1+\gamma e^{-y_1})(1+\gamma e^{-y_2})}} \le \frac{1}{1 + \gamma e^{-\bar{y}}}, 
\]
where the last inequality follows from the AM–GM inequality by 
\[
(1+\gamma e^{-y_1})(1+\gamma e^{-y_2}) = 1 + \gamma (e^{-y_1} + e^{-y_2}) + \gamma^2 e^{-2\bar{y}} \ge 1 + 2\gamma e^{-\bar{y}} + \gamma^2 e^{-2\bar{y}} 
= (1 + \gamma e^{-\bar{y}})^2. 
\]
Since $y_i \in J_{d_i}$ for $i=1,2$, we have
\[
e^{\bar{y}} 
= \sqrt{e^{y_1} \cdot e^{y_2}} 
\le \sqrt{\frac{\lambda}{\gamma^{d_1}} \cdot \frac{\lambda}{\gamma^{d_2}}} 
= \frac{\lambda}{\gamma^{\bar{d}}}. 
\]
If $\bar{d} \ge 2$, then we deduce from \cref{lem:twocriticalthresholdbounds} and $\gamma > 1$ that 
\[
e^{\bar{y}} 
\le \frac{\lambda_c}{\gamma^{\lfloor \bar{d} \rfloor}} 
\le \frac{4\gamma}{\lfloor \bar{d} \rfloor - 1}. 
\]
It follows that
\[
\sqrt{|h(y_1)| \cdot |h(y_2)|} \le \frac{1}{1 + \gamma e^{-\bar{y}}} \le \frac{1}{1 + \frac{\lfloor \bar{d} \rfloor - 1}{4}} = \frac{4}{\lfloor \bar{d} \rfloor + 3} \le \frac{8}{d_1+d_2 + 2}. 
\]
If $\bar{d} < 2$, then it is easy to see that 
\[
\sqrt{|h(y_1)| \cdot |h(y_2)|} \le 1 \le \frac{8}{d_1+d_2+2}. 
\]


\medskip
\noindent
\emph{S.~~Soft-constraint models: $\beta > 0$ and $\lambda \in \mathcal{A}$.}

\medskip
\noindent
\emph{S.1.~~$\sqrt{\beta \gamma} > \frac{\Delta-2}{\Delta}$.} 

\medskip
For every $y_1 \in J$ we deduce from \cref{lem:psimaximizer} that 
\[
|h(y_1)| \le \frac{1-\sqrt{\beta\gamma}}{1+\sqrt{\beta\gamma}} \le \frac{1}{\Delta-1} \le \frac{1.5}{\Delta}. 
\]


\medskip
\noindent
\emph{S.2.~~$\sqrt{\beta \gamma} \le \frac{\Delta-2}{\Delta}$ and $\gamma \le 1$.}

\medskip
In this case, we have either $\lambda < \lambda_c$ or $\lambda > \overline{\lambda}_c$ where $\lambda_c, \overline{\lambda}_c$ are the two critical fields. 
Consider first $\lambda > \overline{\lambda}_c$. 
For every $y_1 \in J_{d_1}$ we deduce from \cref{lem:twocriticalthresholdbounds} and $\beta < 1$ that 
\[
e^{y_1} \ge \lambda \beta^{d_1} \ge \overline{\lambda}_c \beta^{\Delta-1} \ge \frac{\theta(\Delta-1)}{18\beta}
\]
where $\theta(d) = d(1-\beta\gamma) - (1+\beta\gamma)$. Hence, 
\begin{align*}
|h(y_1)| = \frac{(1-\beta\gamma)e^{y_1}}{(\beta e^{y_1} + 1)(e^{y_1} + \gamma)} 
&= \frac{1-\beta\gamma}{\beta e^{y_1} + \gamma e^{-y_1} + (1+\beta\gamma)}\\ 
&\le \frac{1-\beta\gamma}{\frac{\theta(\Delta-1)}{18} + (1+\beta\gamma)} = \frac{18(1-\beta\gamma)}{(\Delta-1)(1-\beta\gamma) + 17(1+\beta\gamma)} 
\le \frac{18}{\Delta}. 
\end{align*}


Next we consider $\lambda < \lambda_c$. 
For every $y_1 \in J_{d_1}$ we deduce from \cref{lem:twocriticalthresholdbounds} and $\gamma \le 1$ that 
\[
e^{y_1} \le \frac{\lambda}{\gamma^{d_1}} \le \frac{\lambda_c}{\gamma^{\Delta-1}} \le \frac{18\gamma}{\theta(\Delta-1)}. 
\]
Hence, 
\[
|h(y_1)| 
= \frac{1-\beta\gamma}{\beta e^{y_1} + \gamma e^{-y_1} + (1+\beta\gamma)} 
\le \frac{1-\beta\gamma}{\frac{\theta(\Delta-1)}{18} + (1+\beta\gamma)} \le \frac{18}{\Delta}. 
\]


\medskip
\noindent
\emph{S.3.~~$\sqrt{\beta \gamma} \le \frac{\Delta-2}{\Delta}$ and $\gamma > 1$.} 

\medskip
Let $\bar{y} = \frac{y_1+y_2}{2}$, $\bar{d} = \frac{d_1+d_2}{2}$, 
$d_L = \lfloor \bar{d} \rfloor$, and $d_R = \lceil \bar{d} \rceil$. 
We first consider some trivial cases. 
If $\bar{d} \le 2$ then it is easy to see that
\[
\sqrt{|h(y_1)| \cdot |h(y_2)|} \le 1 \le \frac{6}{d_1+d_2+2}.
\]
If $\bar{d} > 2$ and $d_L \le \overline{\Delta}$, then we deduce from \cref{lem:psimaximizer} that 
\[
\sqrt{|h(y_1)| \cdot |h(y_2)|} 
\le \frac{1-\sqrt{\beta\gamma}}{1+\sqrt{\beta\gamma}} 
= \frac{1}{\overline{\Delta}} 
\le \frac{2}{d_1+d_2-2} 
\le \frac{6}{d_1+d_2+2}. 
\]
Hence, in the following we may assume that $\bar{d} > 2$ and $d_L > \overline{\Delta}$.

Since the parameters $(\beta,\gamma,\lambda)$ are up-to-$\Delta$ unique, we have $\lambda \in \mathcal{A}$ where the regime $\mathcal{A}$ is given by \cref{eq:regime-uniq}.
Observe that
\[
\mathcal{A} \subseteq (0, \lambda_1(d_L)) \cup (\lambda_2(d_R), \infty) \cup (\lambda_2(d_L), \lambda_1(d_R))
\]
where the last interval is nonempty only when $\lambda_2(d_L) < \lambda_1(d_R)$. 
This means that $\lambda$ is contained in at least one of the three intervals.
We establish the bound by considering these three cases separately.


\emph{Case 1: $\lambda < \lambda_1(d_L)$.}
By the Cauchy-Schwarz inequality, we have
\begin{align}
\sqrt{|h(y_1)| \cdot |h(y_2)|} 
&= \sqrt{ \frac{1-\beta\gamma}{\beta e^{y_1} + \gamma e^{-y_1} + (1+\beta\gamma)} } \cdot \sqrt{ \frac{1-\beta\gamma}{\beta e^{y_2} + \gamma e^{-y_2} + (1+\beta\gamma)} } \nonumber\\ 
&\le \frac{1-\beta\gamma}{\sqrt{ (\beta e^{y_1} + \gamma e^{-y_1}) (\beta e^{y_2} + \gamma e^{-y_2}) } + (1+\beta\gamma)}. \label{eq:Cauchy-Schwarz}
\end{align} 
Therefore, we get
\[
\sqrt{|h(y_1)| \cdot |h(y_2)|} \le \frac{1-\beta\gamma}{\gamma e^{-\bar{y}} + (1+\beta\gamma)}.
\]
Since $y_i \in J_{d_i}$ for $i=1,2$ and $\gamma > 1$, we deduce from \cref{lem:twocriticalthresholdbounds} that
\[
e^{\bar{y}} 
\le \frac{\lambda}{\gamma^{\bar{d}}} 
\le \frac{\lambda_1(d_L)}{\gamma^{d_L}}
\le \frac{18\gamma}{\theta(d_L)},
\]
where $\theta(d_L) = d_L(1-\beta\gamma) - (1+\beta\gamma)$. 
It follows that
\[
\sqrt{|h(y_1)| \cdot |h(y_2)|} \le \frac{1-\beta\gamma}{\gamma e^{-\bar{y}} + (1+\beta\gamma)} 
\le \frac{1-\beta\gamma}{\frac{\theta(d_L)}{18} + (1+\beta\gamma)} 
\le \frac{36}{d_1+d_2 + 2}. 
\]


\emph{Case 2: $\lambda > \lambda_2(d_R)$.}
Similarly, we obtain from \cref{eq:Cauchy-Schwarz} that
\[
\sqrt{|h(y_1)| \cdot |h(y_2)|} \le \frac{1-\beta\gamma}{\beta e^{\bar{y}} + (1+\beta\gamma)}.
\]
Since $y_i \in J_{d_i}$ for $i=1,2$ and $\beta < 1$, we deduce from \cref{lem:twocriticalthresholdbounds} that
\[
e^{\bar{y}} 
\ge \lambda \beta^{\bar{d}}
\ge \lambda_2(d_R) \beta^{d_R}
\ge \frac{\theta(d_R)}{18\beta},
\]
where $\theta(d_R) = d_R(1-\beta\gamma) - (1+\beta\gamma)$. 
It follows that
\[
\sqrt{|h(y_1)| \cdot |h(y_2)|} \le \frac{1-\beta\gamma}{\beta e^{\bar{y}} + (1+\beta\gamma)} \le \frac{1-\beta\gamma}{\frac{\theta(d_R)}{18} + (1+\beta\gamma)} 
\le \frac{36}{d_1+d_2 + 2}. 
\]

\emph{Case 3: $\lambda_2(d_L) <\lambda< \lambda_1(d_R)$.}
We may assume that $d_1 \ge d_2$. 
By \cref{eq:Cauchy-Schwarz}, we obtain
\[
\sqrt{|h(y_1)| \cdot |h(y_2)|} 
\le \frac{1-\beta\gamma}{\sqrt{\beta\gamma} e^{\frac{y_2-y_1}{2}} + (1+\beta\gamma)}. 
\]
Since $y_i \in J_{d_i}$ for $i=1,2$ and $\beta < 1 < \gamma$, we have
\[
e^{y_2-y_1} \ge \beta^{d_2}\gamma^{d_1} \ge \beta^{d_L}\gamma^{d_R}.
\]
Meanwhile, we deduce from \cref{lem:twocriticalthresholdbounds} that
\[
\frac{\theta(d_L)}{18 \beta^{d_L + 1}} \le \lambda_2(d_L) < \lambda < \lambda_1(d_R) \le \frac{18\gamma^{d_R + 1}}{\theta(d_R)}, 
\] 
which implies
\[
\sqrt{\beta\gamma} e^{\frac{y_2-y_1}{2}} \ge \sqrt{\beta^{d_L+1} \gamma^{d_R+1}} \ge \frac{\sqrt{\theta(d_L) \theta(d_R)}}{18} \ge \frac{\theta(d_L)}{18}. 
\]
It follows that
\[
\sqrt{|h(y_1)| \cdot |h(y_2)|} 
\le \frac{1-\beta\gamma}{\sqrt{\beta\gamma} e^{\frac{y_2-y_1}{2}} + (1+\beta\gamma)}
\le \frac{1-\beta\gamma}{\frac{\theta(d_L)}{18} + (1+\beta\gamma)}
\le \frac{36}{d_1+d_2 + 2}. \qedhere
\]
\end{proof}

\subsection{Proofs of technical lemmas}
\label{sec:boundedness-lemma}

\begin{proof}[Proof of \cref{lem:twocriticalthresholdbounds}]
1. 
For every $1<d<\Delta$ we have
\[
\lambda_c \le \frac{\gamma^{d+1} d^d}{(d-1)^{d+1}} = \frac{\gamma^{d+1}}{d-1} \left( \frac{d}{d-1} \right)^d \le \frac{4\gamma^{d+1}}{d-1},
\]
where the last inequality follows from that $(\frac{d}{d-1})^d \le 4$ for all integer $d > 1$. 

\medskip
\noindent 2. For every $\overline{\Delta} \le d <\Delta$ we have 
\[
x_1(d) = \frac{2\gamma}{\theta(d) + \sqrt{\theta(d)^2 - 4\beta\gamma}} \le \frac{2\gamma}{\theta(d)}. 
\]
Observe that the function $\frac{x+\gamma}{\beta x+1}$ is monotone increasing in $x$ when $\beta \gamma <1$, and thus we deduce that
\[
\frac{x_1(d) + \gamma}{\beta x_1(d) + 1} \le \frac{\frac{2\gamma}{\theta(d)} + \gamma}{\frac{2\beta\gamma}{\theta(d)} + 1} 
= \gamma \cdot \frac{2 + d(1-\beta\gamma) - (1+\beta\gamma)}{2\beta\gamma + d(1-\beta\gamma) - (1+\beta\gamma)} 
= \gamma \cdot \frac{d+1}{d-1}. 
\]
Therefore, 
\[
\lambda_1(d) = x_1(d) \left( \frac{x_1(d) + \gamma}{\beta x_1(d) + 1} \right)^d \le \frac{2\gamma}{\theta(d)} \cdot \gamma^d \cdot \left( \frac{d+1}{d-1} \right)^d \le \frac{18\gamma^{d+1}}{\theta(d)}
\]
where the last inequality follows from that $(\frac{d+1}{d-1})^d \le 9$ for all integer $d > 1$. 

The second part can be proved similarly. For every $\overline{\Delta} \le d <\Delta$ we have 
\[
x_2(d) = \frac{\theta(d) + \sqrt{\theta(d)^2 - 4\beta\gamma}}{2\beta} \ge \frac{\theta(d)}{2\beta}, 
\]
and hence,
\[
\frac{x_2(d) + \gamma}{\beta x_2(d) + 1} \ge \frac{\frac{\theta(d)}{2\beta} + \gamma}{ \frac{\theta(d)}{2} + 1} 
= \frac{1}{\beta} \cdot \frac{d(1-\beta\gamma) - (1+\beta\gamma) + 2\beta \gamma}{d(1-\beta\gamma) - (1+\beta\gamma) + 2} 
= \frac{1}{\beta} \cdot \frac{d-1}{d+1}. 
\]
We then conclude that
\[
\lambda_2(d) = x_2(d) \left( \frac{x_2(d) + \gamma}{\beta x_2(d) + 1} \right)^d \ge \frac{\theta(d)}{2\beta} \cdot \frac{1}{\beta^d} \cdot \left( \frac{d-1}{d+1} \right)^d \ge \frac{\theta(d)}{18 \beta^{d+1}}, 
\]
where the last inequality again follows from that $(\frac{d+1}{d-1})^d \le 9$ for all integer $d > 1$. 
\end{proof}

\begin{proof}[Proof of \cref{lem:psimaximizer}]
We deduce from the AM–GM inequality that 
\[
|h(y)| = \frac{|1-\beta\gamma|}{\beta e^y + \gamma e^{-y} + 1+\beta} \le \frac{|1-\beta\gamma|}{2\sqrt{\beta\gamma} + 1+\beta} = \frac{|1-\sqrt{\beta\gamma}|}{1+\sqrt{\beta\gamma}}. \qedhere
\]
\end{proof}



\section{Proofs for ferromagnetic cases}\label{sec:ferroproofs}
\subsection{Proof of \texorpdfstring{\cref{thm:ss19ferroregion}}{Theorem 26}}\label{sec:ferro_proof_SS19}

\begin{proof}[Proof of \cref{thm:ss19ferroregion}]
Throughout, we use the ``trivial potential'' function $\Psi(y) = y$. Note that then, $\psi(y) = 1$ is a constant function. 
Now, we prove \hyperlink{cond:contraction}{Contraction} and \hyperlink{cond:boundedness}{Boundedness}. We split into the three cases.
\begin{enumerate}
    \item 
    We first prove the \hyperlink{cond:contraction}{Contraction} part. 
    By \cref{lem:psimaximizer}, for all $y\in[-\infty,+\infty]$ we have
    \[
    \abs{h(y)} 
    \leq \frac{|1 - \sqrt{\beta\gamma}|}{1 + \sqrt{\beta\gamma}} \leq \frac{1-\delta}{\Delta-1}.
    \]
    Now let us prove the \hyperlink{cond:boundedness}{Boundedness} condition. From the above inequality we have
    \[
    \abs{h(y)} \le \frac{1}{\Delta-1} \le \frac{1.5}{\Delta} 
    \]
    for $\Delta \ge 3$. 

    \item For the \hyperlink{cond:contraction}{Contraction} part, since $\log(\lambda \max\{1,1/\gamma^{\Delta-1}\}) \leq y_{i} \leq \log(\lambda \max\{1,\beta^{\Delta-1}\})$, we have
    \begin{align*}
        \abs{\frac{\partial H_{d}(\mybf{y})}{\partial y_{i}}} &= \abs{h(y_{i})} = \frac{\beta\gamma-1}{1 + \beta\gamma + \gamma e^{-y_{i}} + \beta e^{y_{i}}} \leq \frac{\beta\gamma-1}{1 + \beta\gamma + \gamma e^{-y_{i}}} \\
        &\leq \frac{\beta \gamma - 1}{1 + \beta\gamma + \frac{\gamma}{\lambda \max\{1,\beta^{\Delta-1}\}}}.
    \end{align*}
    Since we assumed $\lambda \leq (1 - \delta)\frac{\gamma}{\max\{1,\beta^{\Delta-1}\} \cdot ((\Delta-2)\beta\gamma - \Delta)}$, it follows that we have the upper bound
    \begin{align*}
        \frac{\beta\gamma - 1}{1 + \beta\gamma + \frac{(\Delta-2)\beta\gamma - \Delta}{1 - \delta}} 
        &= (1 - \delta) \frac{\beta \gamma - 1}{(\Delta-1-\delta)\beta\gamma - (\Delta - 1 + \delta)}\\ 
        &= (1 - \delta) \frac{\beta\gamma-1}{(\Delta-1-\delta)(\beta\gamma - 1) + 2\delta} \\
        &\leq \frac{1 - \delta}{\Delta - 1 - \delta} \leq (1 - \Theta(\delta)) \frac{1}{\Delta-1}.
    \end{align*}
    Now, we prove the \hyperlink{cond:boundedness}{Boundedness} condition. Note that since $\lambda \leq \frac{\gamma}{\max\{1,\beta^{\Delta-1}\} \cdot ((\Delta-2)\beta\gamma-\Delta)}$, it follows that $y \leq \log(\lambda \max\{1,\beta^{\Delta-1}\}) \leq \log\wrapp{\frac{\gamma}{(\Delta-2)\beta\gamma - \Delta}}$. 
    A simple calculation reveals that $\frac{\gamma}{(\Delta-2)\beta\gamma - \Delta} \leq \sqrt{\frac{\gamma}{\beta}}$ and so by \cref{lem:psimaximizer}, we have
    \begin{align*}
        \abs{h(y)} &\leq \abs{h\wrapp{\log\wrapp{\frac{\gamma}{(\Delta-2)\beta\gamma - \Delta}}}} \leq \frac{(\beta\gamma-1)e^{\log\wrapp{\frac{\gamma}{(\Delta-2)\beta\gamma - \Delta}}}}{e^{\log\wrapp{\frac{\gamma}{(\Delta-2)\beta\gamma - \Delta}}} + \gamma} \\
        &= (\beta\gamma-1) \frac{1}{1 + (\Delta-2)\beta\gamma - \Delta} = \frac{\beta\gamma-1}{(\Delta-2)(\beta\gamma - 1) - 1} \leq O(1/\Delta).
    \end{align*}

    \item For the \hyperlink{cond:contraction}{Contraction} part, since $\log(\lambda \max\{1,1/\gamma^{\Delta-1}\}) \leq y_{i} \leq \log(\lambda \max\{1,\beta^{\Delta-1}\})$, we have
    \begin{align*}
        \abs{\frac{\partial H_{d}(\mybf{y})}{\partial y_{i}}} &= \abs{h(y_{i})} = \frac{\beta\gamma-1}{1 + \beta\gamma + \gamma e^{-y_{i}} + \beta e^{y_{i}}} \leq \frac{\beta\gamma-1}{1 + \beta\gamma + \beta e^{y_{i}}} \\
        &\leq \frac{\beta \gamma - 1}{1 + \beta\gamma + \beta \lambda \max\{1,1/\gamma^{\Delta-1}\}}.
    \end{align*}
    Since we assumed $\lambda \geq \frac{1}{1 - \delta} \cdot \frac{(\Delta-2)\beta\gamma-\Delta}{\beta \cdot \min\{1,1/\gamma^{\Delta-1}\}}$, it follows that we have the upper bound
    \begin{align*}
        \frac{\beta\gamma - 1}{1 + \beta\gamma + \frac{(\Delta-2)\beta\gamma - \Delta}{1 - \delta}}
    \end{align*}
    which is again is upper bounded by $(1 - \Theta(\delta))\frac{1}{\Delta-1}$ as we calculated in case 2 above.
    
    Now, we prove the \hyperlink{cond:boundedness}{Boundedness} condition. Note that since $\lambda \geq \frac{(\Delta-2)\beta\gamma - \Delta}{\beta \min\{1,1/\gamma^{\Delta-2}}$, it follows that $y \geq \log(\lambda \min\{1,1/\gamma^{\Delta-1}\} \geq \log\wrapp{\frac{(\Delta-2)\beta\gamma - \Delta}{\beta}}$. 
    A simple calculation reveals that $\frac{(\Delta-2)\beta\gamma - \Delta}{\beta} \geq \sqrt{\frac{\gamma}{\beta}}$ and so by \cref{lem:psimaximizer}, we have
    \begin{align*}
        \abs{h(y)} &\leq \abs{h\wrapp{\log\wrapp{\frac{(\Delta-2)\beta\gamma - \Delta}{\beta}}}} \leq (\beta \gamma - 1) \frac{1}{\beta \cdot \frac{(\Delta-2)\beta\gamma - \Delta}{\beta} + 1} \\
        &= \frac{\beta\gamma - 1}{(\Delta-2)(\beta\gamma - 1) - 1} \leq O(1/\Delta).\qedhere
    \end{align*}
\end{enumerate}
\end{proof}

\subsection{Proof of \texorpdfstring{\cref{thm:gl18ferroregion}}{Theorem 27}}
In this subsection, we use results from \cite{GL18} to prove \cref{thm:gl18ferroregion}. Their potential function is implicitly defined by its derivative for the marginal ratios as
\begin{align*}
    \Phi'(R) = \phi(R) = \min\wrapc{\frac{\beta\gamma-1}{\alpha\gamma \log \frac{\lambda + \gamma}{\beta\lambda+1}}, \frac{1}{R \log \frac{\lambda}{R}}}
\end{align*}
for a constant $0 \leq \alpha \leq 1$ depending only on $\beta,\gamma,\lambda$ (see \cite{GL18} for a precise definition). In our context, the corresponding potential for the log ratios is
\begin{align*}
    \Psi'(y) = \psi(y) = e^{y}\phi(e^{y}) = \min\wrapc{\frac{\beta\gamma-1}{\alpha\gamma \log \frac{\lambda + \gamma}{\beta\lambda+1}} e^{y}, \frac{1}{\log \frac{\lambda}{e^{y}}}}
\end{align*}
and is bounded by constants depending on $\beta,\gamma,\lambda,\Delta$ for $\log(\lambda/\gamma^{\Delta-1}) \leq y \leq \log \lambda$.

One of the main technical results in \cite{GL18} is showing that the tree recursion is contracting with the potential function $\Phi$, and the derivative $\phi$ is bounded in the sense that there exist positive constants $C_{1},C_{2}$ depending only on $\beta,\gamma,\lambda$ such that $C_{1} \leq \phi(R) \leq C_{2}$ for all $0 \leq R \leq \lambda$. \cite{GL18} refer to such a function as a \textit{universal potential function}.

In our context, we get that $\Psi$ is an $(\alpha,c)$-potential function which satisfies \cref{defn:potential}, but with a constant $c$ that depends on $\gamma,\Delta$. Indeed, worst case, we have
\begin{align*}
    \max_{y_{1},y_{2}} \frac{\psi(y_{2})}{\psi(y_{1})} \geq \frac{\psi(\log \lambda)}{\psi(\log(\lambda/\gamma^{\Delta-1}))} = \frac{\lambda\frac{\beta\gamma-1}{\alpha \gamma \log\frac{\lambda+\gamma}{\beta\lambda+1}}}{\frac{\beta\gamma-1}{\alpha \log \frac{\lambda+\gamma}{\beta\lambda+1}} \cdot \frac{\lambda}{\gamma^{\Delta}}} = \gamma^{\Delta-1}.
\end{align*}
More precisely, we have the following result from \cite{GL18}, stated in terms of the log marginal ratios.
\begin{thm}
\label{thm:gl18universalpotential}
Assume $\beta,\gamma,\lambda$ are nonnegative real numbers satisfying $\beta \leq 1 \leq \gamma$, $\sqrt{\beta\gamma} \geq 1$, and $\lambda < \wrapp{\frac{\gamma}{\beta}}^{\frac{\sqrt{\beta\gamma}}{\sqrt{\beta\gamma} - 1}}$. Then the function $\Psi$ is an $(\alpha,c)$-potential function for a constant $0 < \alpha < 1$ depending on $\beta,\gamma,\lambda$, and a constant $c > 0$ depending on $\beta,\gamma,\lambda,\Delta$.
\end{thm}
Combined with \cref{thm:contraction-implies-mixing}, this gives $O(n^{C})$ mixing with a constant $C$ depending only on $\beta,\gamma,\lambda,\Delta$. We note this is weaker than the correlation decay result in \cite{GL18}, since there, $C$ does not depend on $\Delta$, and hence is efficient for arbitrary graphs.

\section{Slightly faster mixing}\label{sec:slightlyfastermixing}
In this section, we slightly optimize our mixing time results for certain antiferromagnetic 2-spin systems by more carefully taking into account the tradeoff between the (nontrivial) spectral independence bound we prove based on contraction, and the (trivial) spectral independence bound we obtained in \cref{subsec:constantdimensions} for handling constant-sized graphs.
\begin{prop}
Suppose a distribution $\mu$ on subsets of $[n]$ is $(\eta_{0},\dots,\eta_{n-2})$-spectrally independent for $\eta_{i} \leq \min\{a, (n-i-1)b\}$, for some $a \geq 0$ and $0 \leq b \leq 1$. Then the Glauber dynamics for sampling from $\mu$ has spectral gap at least $\frac{1}{n} \cdot \Omega\wrapp{\frac{a}{b n}}^{a}$
\end{prop}
\begin{proof}
Suppose we have already conditioned on $c$-fraction of elements to be ``in/out. The resulting distribution is both $b (1-c)n$-spectrally independent and $a$-spectrally independent. The exact threshold $c$ for which the bound $b(1-c)n$ is better than $a$ is given by
\begin{align*}
    c = 1 - \frac{a}{b n}
\end{align*}
We note such a $c$ only makes sense when $0 \leq 1 - \frac{a}{b n} \leq 1$, or equivalently, $b n \geq a$. Now, we apply the $a$-spectral independence bound for all conditional distributions based on fixing at most $c$-fraction of vertices. We apply the $(n-i-1)b$-spectral independence otherwise. We obtain a final spectral gap lower bound of
\begin{align*}
    \frac{1}{n} \cdot (1 - b)^{(1 - c)n} \cdot \prod_{k=0}^{cn} \wrapp{1 - \frac{a}{n-k-1}}
\end{align*}
Observe that
\begin{align*}
    (1 - b)^{(1-c)n} = (1 - b)^{\frac{a}{b}} \gtrsim \exp(-a)
\end{align*}
We also have
\begin{align*}
    \prod_{k=0}^{cn} \wrapp{1 - \frac{a}{n-k-1}} &\gtrsim \exp\wrapp{-a \sum_{k=0}^{cn} \frac{1}{n-k-1}} \\
    &\gtrsim \exp\wrapp{-a \wrapp{\underset{\approx \log n}{\underbrace{\sum_{k=0}^{n-2} \frac{1}{n-k-1}}} - \underset{\approx \log (1-c)n}{\underbrace{\sum_{k=cn+1}^{n-2} \frac{1}{n-k-1}}}}} \\
    &\gtrsim \exp\wrapp{-a \cdot \log \frac{1}{1-c}} \\
    &\gtrsim \exp\wrapp{-a \log \frac{b n}{a}} \\
    &\gtrsim \wrapp{\frac{a}{b n}}^{a}
\end{align*}
Putting these together, we obtain the desired lower bound.
\end{proof}
With this result, we can apply it to the antiferromagnetic models with $\sqrt{\beta\gamma} \leq \frac{\Delta-2}{\Delta}, \gamma \leq 1$ and $\beta = 0, \gamma \leq 1$, since looking in the proof of \cref{lem:antiferroconstantsize}, we have such systems are $Cn$-spectrally independent roughly with $C \leq O(1/\Delta)$.
\begin{coro}[Soft Constraints]
Fix integers $\Delta \geq 3$, $1 < \overline{\Delta} < \Delta$. Let $\beta,\gamma,\lambda \geq 0$ be nonnegative real numbers satisfying $\frac{\overline{\Delta}-2}{\overline{\Delta}} \leq \sqrt{\beta\gamma} \leq \frac{\overline{\Delta}-1}{\overline{\Delta}+1}$ and $\gamma \leq 1$. Assume further that $(\beta,\gamma,\lambda)$ is up-to-$\Delta$ unique with gap $0 < \delta < 1$. Then for every $n$-vertex graph $G$ with maximum degree at most $\Delta$, the Glauber dynamics for sampling from the antiferromagnetic 2-spin system with parameters $(\beta,\gamma,\lambda)$ mixes in $O\wrapp{\frac{\overline{\Delta} \cdot n}{\Delta}}^{O(1/\delta)}$ steps.
\end{coro}
\begin{coro}[Hard Constraints]
Fix an integer $\Delta \geq 3$, fix $\beta = 0$, and let $0 \leq \gamma \leq 1, \lambda \geq 0$ be up-to-$\Delta$ unique with gap $0 < \delta < 1$. Then for every $n$-vertex graph $G$ with maximum degree at most $\Delta$, the Glauber dynamics for sampling from the antiferromagnetic 2-spin system with parameters $(\beta,\gamma,\lambda)$-mixes in $O\wrapp{\frac{n}{\Delta}}^{O(1/\delta)}$ steps.
\end{coro}

\end{appendices}
\end{document}